\documentclass[journal,draftcls,onecolumn,12pt]{IEEEtran}
\usepackage{graphicx}
\usepackage[printonlyused]{acronym}
\usepackage{bbm,dsfont,mathrsfs,fixmath}
\usepackage{stmaryrd}
\usepackage{amsmath,amstext,amsfonts,amssymb}
\usepackage{epsfig,float}
\usepackage{upgreek}
\usepackage{graphics}
\usepackage{dsfont}
\usepackage{psfrag}
\usepackage{caption}
\usepackage{subcaption}
\usepackage{url}
\usepackage{shadow,color,pifont,times,rotate}

\newcommand{\mb}{\mathbf}

 \DeclareMathAlphabet{\mathpzc}{OT1}{pzc}{m}{it}
\newtheorem{theorem}{Theorem}

\newtheorem{proposition}{Proposition}
\newtheorem{corollary}{Corollary}
\newtheorem{remark}{Remark}
\newtheorem{definition}{Definition}
\newenvironment{proof}[1][Proof:]{\begin{trivlist}
\item[\hskip \labelsep {\itshape #1}]}{\end{trivlist}}
\usepackage{enumerate}    
\ifCLASSINFOpdf
		\usepackage{epstopdf}
\else
\fi
\usepackage{enumerate}    
\DeclareRobustCommand{\prob}[1][P]{\ensuremath {\mathbb{#1}}}

\def\esp{{\mathbb E}}  
\usepackage{mathtools}   

\DeclarePairedDelimiter{\norm}{\lVert}{\rVert}

\usepackage{array}
\usepackage{pgf,tikz}
\usetikzlibrary{arrows}
\usepackage{pgfplotstable}
\usepackage{rotate}

\begin{document}
%

\acrodef{RSS}{Received Signal Strength}
\acrodef{RSSI}{Received Signal Strength Indicator}
\acrodef{KL}{Kullback-Leibler}
\acrodef{RF}{Radio-Frequency}
\acrodef{UWB}{Ultra-Wide Band}
\acrodef{ToA}{Time of Arrival}
\acrodef{AoA}{Angle of Arrival}
\acrodef{kNN}{k-Nearest Neighbors}
\acrodef{LQI}{Link Quality Information}
\acrodef{BER}{Bit-Error Ratio}
\acrodef{AP}{Access Point}
\acrodef{MAC}{Media Access Control}
\acrodef{FPS}{Fingerprinting Localization Algorithms} 
\acrodef{HT}{hypothesis testing} 

\date{}
\title{ {A Mathematical Model for Fingerprinting-based Localization Algorithms\vspace{-1mm}}}

\author{Arash~Behboodi,~\IEEEmembership{Member,~IEEE,}
         Filip Lemic,~\IEEEmembership{Student Member,~IEEE,} \\
        and~Adam~Wolisz,~\IEEEmembership{Senior Member,~IEEE}
\IEEEcompsocitemizethanks{\IEEEcompsocthanksitem Arash Behboodi is with the Institute for Theoretical Information Technology, RWTH Aachen University, Germany.
\IEEEcompsocthanksitem Filip Lemic and Adam Wolisz are with Telecommunication Networks Group (TKN), Technische Universit\"at Berlin, Germany.}

}

\IEEEtitleabstractindextext{%
\begin{abstract}
A general theoretical framework for \ac{FPS}, given their popularity, can be utilized for their performance studies. In this work, after setting up an abstract model for \ac{FPS}, it is shown that fingerprinting-based localization problem can be cast as a \ac{HT} problem and therefore various results in \ac{HT} literature can be used to provide insights, guidelines and performance bounds for general \ac{FPS}. This framework results in characterization of scaling limits of localization reliability in terms of number of measurements and other environmental parameters. It is suggested that \ac{KL} divergence between probability distributions of selected feature for fingerprinting at  different locations encapsulates  information about both accuracy and latency and can be used as a central performance metric for studying \ac{FPS}. Although developed for an arbitrary fingerprint, the framework is particularly used for studying simple \ac{RSS}-based algorithm. The effect of various parameters on the performance of fingerprinting algorithms is discussed, which includes path loss and fading characteristics, number of measurements at each point, number of anchors and their position, and placement of training points. Representative simulations and experimentation are used to verify validity of the theoretical frameworks in realistic setups. 
\end{abstract}
}
\maketitle

\section{Introduction}
\label{introduction}

Precise location of people, equipments, and materials, both indoor and outdoor, is an essential information for future networks as an enabler of context aware services, location aware and pervasive computing, ambient intelligence, and location based services. Variety of localization solutions have been studied and proposed by researchers during years. 
These solutions exploit the spatial dependence of certain properties. For instance, the propagation model of waves in space can be a basis for deriving certain wave characteristics in different locations and then devising algorithms for distinguishing different locations using those characteristics. Ultrasonic, infrared, and \ac{RF} waves are some of the candidates for localization \cite{medina_ultrasound_2013, brassart_localization_2000, erol-kantarci_survey_2011, amundson_survey_2009}, although the later is more popular due to its low cost and availability. \ac{RF}-based solutions leverage available technologies such as IEEE~802.11 (WiFi), IEEE~802.15.4 (ZigBee), IEEE~802.15 (Bluetooth), \ac{UWB}, RFID, and mobile telephony. Different \ac{RF} characteristics can be used for localization including \ac{ToA}, \ac{AoA},\ac{RSS}, and the quality of the RF transmission in digital communication channels (e.g. \ac{LQI}, \ac{BER}). 
These characteristics are then employed in different localization algorithms. 
Roughly one can distinguish three different categories of localization algorithms namely geometry-based, fingerprinting, and Bayesian-based ~\cite{Seco2009}. \ac{RF} Fingerprinting algorithms are particularly attractive because they rely on available wireless infrastructures and do not require costly set-up of a new infrastructure. 

These algorithms are based on the observation that a combination of simple features of different signals in a given location provides a unique signature for that location. The observed feature at each location is used to construct an identification tag for that location which is called the fingerprint. A database is constructed, explicitly or implicitly, by gathering fingerprints of different locations. A pattern matching algorithm  identifies the location by finding the most similar fingerprints in the database to the reported fingerprint from unknown location. 
The signal features should then satisfy some conditions if they are to be used for reliable identification of a location. Firstly, the signal feature should be robust in the localization space. This means that similar fingerprints should correspond to the nearby locations, otherwise a small error in fingerprinting matching can lead to large geometric error in localization. Secondly, the fingerprint should be stable in time, i.e., it should not change significantly from the database creation time to the localization runtime, otherwise the fingerprint of the unknown location can be significantly different from the recorded fingerprint, thereby leading to huge error. The final location estimation includes similarity analysis of fingerprints and optionally post-processing methods such as \ac{kNN}. 
An attractive feature of fingerprinting is that the analytical characterization of the relation between the signal feature and the observer's location is not required. 
This is usually a daunting task as it involves complex characterization of propagation and solving the wave equation for non-homogeneous environments.  
As a matter of fact, the fingerprinting algorithms provide this functional relation between the signal feature and the location through the table of samples recorded in the training database. 
The training database is nothing but samples of the function relating the location in space to the observed feature in the same location.

\vspace{-1mm}
\subsection{Main Contributions}

Most of the existing theoretical results for fingerprinting localization are specific to the chosen signal feature, for instance \ac{RSSI} fingerprinting, and sometimes they do not provide the qualitative characterization for the impact of different parameters on the algorithm design. In this paper, a connection is established between fingerprinting algorithms and hypothesis testing problem. To the best of our knowledge, it is the first time that this connection has been mentioned. In this framework, the signal feature is not specified and it is only assumed to be a random variable with a probability distribution that varies with the location. In this setting, the problem of localization boils down to detecting the probability distribution underlying the observed measurements, which is indeed a hypothesis testing problem. 
Using this formulation, performance limits of fingerprinting algorithms can be characterized using well known results from hypothesis testing. The framework is applicable for a general signal feature which makes it suitable for variety of scenarios. The contributions of this paper are as follows.
Through an abstract formulation, a general framework for theoretical study of fingerprinting algorithms is provided using hypothesis testing problem.
It is shown that there exists a fingerprinting algorithm achieving accurate localization with arbitrarily large training and measurement phase. Moreover the probability of inaccurate localization decays exponentially with  number of measurements. As a result, \ac{KL} divergence between probability distributions of selected feature for fingerprinting at two different locations is suggested as a central metric that encapsulates both accuracy and latency of fingerprinting localization algorithms. 

The introduced framework is discussed for \ac{RSS}-based algorithms which is the most established and promising instance of fingerprinting algorithms. 
The effects of path loss exponent, fading statistics and anchor selection are discussed using this framework. 
These claims are quantitatively verified through simulation and experimental evaluation.

The paper is organized as follows.  Related works are discussed in Section~\ref{relworks}. In Section~\ref{model}, we provide a system model and a formal definition of fingerprinting algorithms. Section~\ref{limits} gives the theoretical limits of fingerprinting algorithms under general assumption about the system. 
In Section~\ref{rss_fingerprinting}, we focus on a theoretical study of a \ac{RSS}-based fingerprinting algorithm. Finally, the theoretical results are verified through simulation and testbed experiments. 

\vspace{-1mm}
 \vspace{-1mm}
\section{Related Works}
\label{relworks}

There are various works in the literature studying different aspects of fingerprinting algorithms~\cite{Milioris2014}. 
A survey of WiFi-based fingerprinting algorithms is provided in~\cite{honkavirta2009comparative} and variants of such algorithms are instantiated to low-dimensional \ac{RSS}~\cite{milioris2011low}, neural network-based clustering~\cite{Laoudias2009}, $K$-means algorithm, clustering, and complexity reduction~\cite{Bai2013} and spatial signal prediction-based training method~\cite{Steiner2011}.

The impact of the number of \acp{AP} on the performance of fingerprinting algorithms has been studied in~\cite{Machaj2010}, where the authors claim that localization error is increased, when the number of \acp{AP} used for constructing training database differs from number of \acp{AP} used in a localization phase. 
In~\cite{Lin2005}, the authors provide a comparison of fingerprinting algorithms based on accuracy, complexity, robustness, and scalability. 
The main theoretical work on fingerprinting algorithms is presented in~\cite{Kaemarungsi2004}, where the authors provide an analysis of the effect of the number of visible \acp{AP} and radio propagation parameters on the performance of fingerprinting algorithms. They provide guidelines for designing and deploying a fingerprinting algorithm where in particular, it is shown that the algorithm does not require a large number of \acp{AP}. Moreover, the grid used for a training database is chosen according to the application requirements, where more dense grids provide {worse accuracy in terms of detecting a correct cell of the grid,} but finer localization in terms of localization errors. These results are extended to complexity analysis in~\cite{Kaemarungsi2005}. The authors in \cite{wen_fundamental_2015} proposed a probabilistic model for \ac{RSS}-based fingerprinting relating the location to received \ac{RSS}. The performance of fingerprinting algorithms has been then discussed using likelihood based detection algorithms and insights have been provided for fingerprinting design. The scalability of fingerprinting algorithms is discussed in~\cite{Ding2013}, where the authors suggest the scalability improvement by reducing a training database size. In~\cite{Meng2011}, the authors discuss the robustness of fingerprinting to outliers and effects such as shadowing. 
As \ac{RSS} differs across different devices, a robust fingerprinting algorithm is proposed in~\cite{Hossain2013} by taking the difference of \ac{RSS} of two \acp{AP} as fingerprints. 
\cite{Beder2012} proposes a method for estimating different antenna attenuations between different devices, and by considering similarly relative differences between \acp{RSS}. 

\vspace{-1mm}
\section{System Model and Definitions}
\label{model}

In this section, the fingerprinting algorithm is formally defined. The used symbols and notations are given in Table~\ref{table:symbols}. The localization space is a connected region $\mathcal{R}$ in $\mathbb R^d$. In practical scenarios, the dimension $d$ is usually $2$ or $3$. The abstract formulation is intended to keep the framework as general as possible. 

\begin{table}[ht]
\centering  
 \begin{tabular}{|c|l|}

 \hline
  Symbol&Description\\
  \hline
  $\mathcal{R}$& Localization space\\
  $\mathcal{S}$& Signal feature space\\
  $\mb{S}, S$& Signal feature\\
  $\mathcal{X}$& Fingerprint space\\
  $\mb{X},X$& Fingerprint\\  
  $\Lambda$& Training grid\\
  $\mb v$& Training points\\
    $\mb u$& Measurement points\\
      $\mb w$& Anchor points\\
  \hline
  \end{tabular}
\caption{Table of Symbols}
\label{table:symbols}
\vspace{-4mm}
\end{table}

\subsection{Fingerprint}

Fingerprinting algorithms are based on associating a fingerprint to a location, which is later used for the identification of that location. A specific feature of environment is chosen as basis for creating a fingerprint. The term \textit{environment} includes multiplicity of possibilities and does not necessarily refer to a particular infrastructure for localization purpose. The signal feature is denoted by $S$ and  belongs to the feature space $\mathcal{S}$. The signal feature can be a combination of multiple real signals, generated by multiple sources like beacons form different APs.  
$m$ consecutive observations of the signal feature $\mb{S}=(S_1,\dots,S_m)\in\mathcal{S}^m$ are a random vector related to the location $\mb{u}$ through the conditional probability $\prob_{\mb S|\mb u}$. Fingerprints are constructed based on observations $\mb{S}$ for each location. 
\begin{definition}[Fingerprint] A fingerprint creating function $f$ is a mapping $\mathcal{S}^m\to\mathcal{X}^n$ that assigns to observations $\mb S$ an element $\mb X$ called the fingerprint at the location $\mathbf{u}$.
 \end{definition}
Note that fingerprint space $\mathcal{X}$ may be different from $\mathcal{S}$.  For example, if a fingerprint is the Gaussian distribution fitted to the vector of $m$ measured received powers, then the fingerprinting space is the space of probability distribution.
Another example is the \ac{AoA} based fingerprinting.  The signal feature belongs to $\mathcal{S}=[0,2\pi)$. If \ac{AoA} is measured $m$ times and the average value is designated as fingerprint, then $n=1$ and $\mathcal{X}=[0,2\pi)$. However if the fingerprint is the empirical distribution of $m$ measurements, then $\mathcal{X}$ is the space of all probability distributions on $[0,2\pi)$.  
Another issue is the robustness of fingerprints. Suppose that for two locations $\mb u_1$ and $\mb u_2$, the probabilities $\prob_{\mb X|\mb u_1}$ and $\prob_{\mb X|\mb u_2}$ are \textit{close enough}. The closeness of probabilities is measured using a metric $d$. The robustness requirement of fingerprints implies that when $d( \prob_{\mb X|\mb u_1}, \prob_{\mb X|\mb u_2})\leq L $ where $d$ is a metric and $L>0$, then $\mb u_1$ and $\mb u_2$ should also be close which is $\norm{\mb u_1-\mb u_2}\leq s(L)$ where $s(L)\to 0$ whenever $L\to 0$. This can be used a definition of spatial robustness of the fingerprint. The quantitative characterization of stability depends on the metric chosen for measuring \textit{closeness} of probabilities.  As it will be discussed later for \ac{RSS} fingerprinting, the robustness analysis of fingerprints is already a good indication for the accuracy of fingerprinting systems.

Fingerprints should also be stable in time. This means that fingerprints should not change dramatically from the moment they are recorded until the moment they are used for identification. The robustness requirement varies with localization systems. It can be defined in terms of closeness of $\prob^{t_0}_{\mb X|\mb u_1}$ and $\prob^{t_1}_{\mb X|\mb u_2}$ where $t_0$ and $t_1$ are respectively training and measurement times. The notion of robustness and stability are very much related in that the spatial robustness guarantees that small change in the fingerprint in time will not result in significant localization error. In this work, the fingerprints are assumed to be stable in time.

\vspace{-1mm}
\subsection{Fingerprinting Algorithm}
After a signal feature and a fingerprint creating function have been selected, the next step is to design the fingerprinting algorithm. The first step is to construct a database consisting of pairs of locations and their fingerprints. Out of an uncountable set of locations in the localization space, only a finite number of locations can be chosen to construct the database. The training database can be constructed through extensive measurements, through simulation-based radio-map construction \cite{scholl_fast_2012} or the combination of both. This step is called the training phase. The set of training locations, called a training grid, is denoted by $\Lambda\subset \mathbb{R}^d$. 
One can choose an algebraic structure for griding using lattices or a non-uniform grids.  The region of nearest points to each  training location $\mb v\in \Lambda$ is called its Voronoi region denoted by $\mathcal{V}_{\mathbf{v}}$.  
The training locations in $\Lambda$ divide the localization space into regions. 
\begin{remark}
\label{rem:Voronoi}
Consider the Voronoi regions of nearest neighbors defined for the training locations. The closest points to a training point are not necessarily those with the closest fingerprints. 
One can define modified Voronoi-region $\hat{\mathcal{V}}_{\mb v_1}$ of a training point $\mb v_1$ as the set of all points $\mb u$ such that their fingerprints are closer to $\mb v_1$ than any other training points. Modified Voronoi regions are in general different from the  Voronoi region. 
\end{remark}

After choosing training locations, a specific feature of environment is chosen for creating a fingerprint. The training database is created by measuring the signal feature $\mb S$ at the training locations inside the training grid $\Lambda$. At each location $\mathbf{v}$, multiple measurements are performed and fingerprints are constructed accordingly.
The training database $\mathcal{D}$, which consists of pairs like $(\mathbf{v},\mb X)$, is a subset of $\Lambda\times \mathcal{X}^n$.  

The next part consists of finding the location of a target node placed at the location $\mathbf{u}$. The target node measures the selected feature $m'$ times and creates a fingerprint $\mb X_{\mb u}\in\mathcal{X}^{n'}$.  For the rest, it is assumed  that $n'=n$ but in general number of measurements in the training phase and in the localization phase can be different. After acquiring the fingerprint, a pattern matching function $g$ is used to estimate the target node's location $\hat{\mathbf{u}}$ based on the acquired fingerprint $\mb X_{\mb u}$ and fingerprints in the training database. The function $g$ is a mapping from $\mathcal X^n$ to the localization space $\mathcal{R}$. The pattern matching function can be regarded as a composition of multiple functions. For instance a similarity kernel function can first find the most similar fingerprints in the database. Then, one can additionally use $k-$nearest neighbor methods to average between the $k$ closest training locations, rather than declaring single training location. 
The average can be weighted or not, adaptive or non-adaptive~\cite{Lemic14experimental_decomposition_of}. 
In any case, $k-$NN methods will have a set of estimated locations $\Lambda_e$, which is not equal to the original set of training locations $\Lambda$. 

The error is obviously $\|\mathbf{u}-\hat{\mathbf{u}}\|$. 

Based on the previous discussion, we can finally specify what a fingerprinting algorithm is.
\begin{definition} [Fingerprinting algorithm] A fingerprinting algorithm for a localization space $\mathcal{R}$ consists of:
\begin{itemize}
\item A set of training locations $\Lambda$
\item A fingerprint creating function $f: \mathcal{S}^m\to\mathcal{X}^n$ that maps a measured signal feature $\mb S$ to a fingerprint $\mb X$ in $\mathcal{X}^n$
\item A training database $\Lambda\times \mathcal{X}^n$ consisting of pairs of locations and their fingerprints
\item A pattern matching function $g: \mathcal{X}^n\to\mathcal{R}$ reporting the final location by comparing the target node's fingerprint to the ones from the training database. 

\end{itemize}
\end{definition}
\vspace{-2mm}

Figure~\ref{fig:srat} represents the functional block of fingerprinting algorithm starting from constructing training database to localization of an arbitrary point $\mb u$. The advantage of this abstraction, as it can be seen later, is its applicability in various scenarios. 

\begin{figure}[!ht]
\vspace{-3mm}
\centering
\includegraphics[width=\columnwidth]{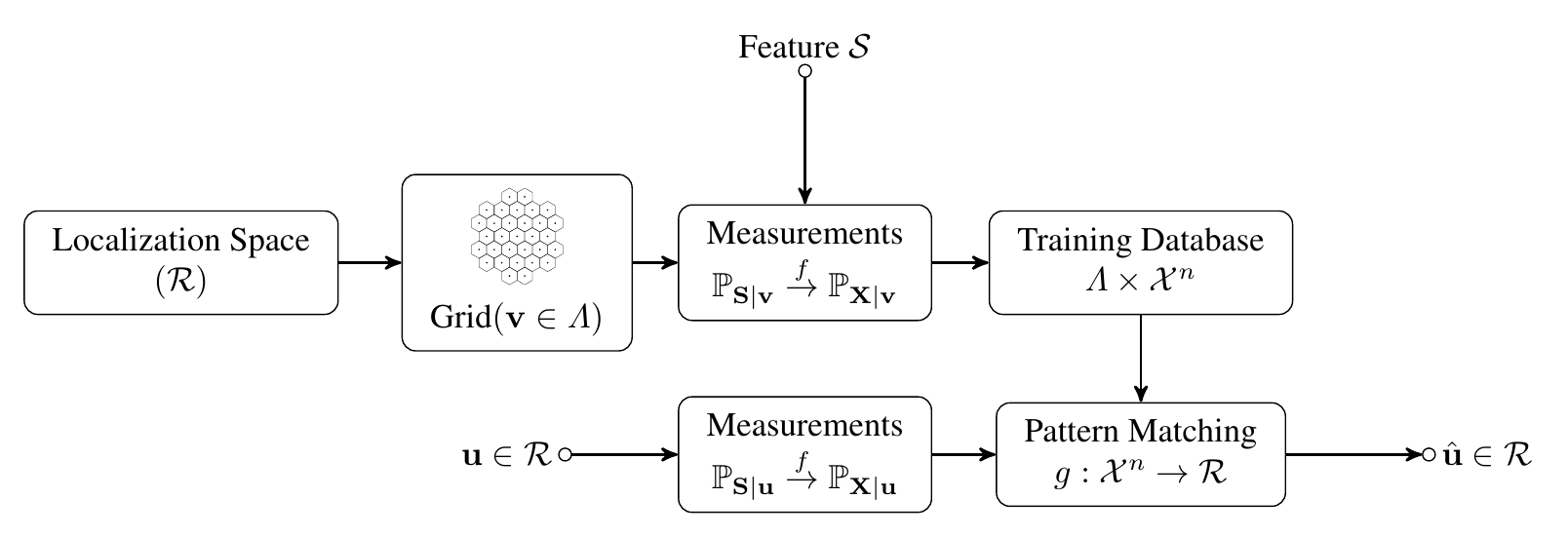}
\vspace{-3mm}
\caption{Functional block of a fingerprinting algorithm} 
\label{fig:srat}
\vspace{-3mm}
\end{figure}

\subsection{Performance of Fingerprinting Algorithms}

After specifying fingerprinting algorithms, we introduce here a framework for evaluating the performance limit of fingerprinting algorithms. 
The main performance metric for any localization algorithm is the localization error. 
 
For a fingerprinting algorithm specified above, the localization error for each $\mathbf{u}\in\mathcal{R}$ is defined as follows:
$$
\Delta(\mathbf{u})=\|\hat{\mathbf{u}}-\mathbf{u}\|.
$$
Similar to the definition of error in information theory and statistics, it is possible to define the maximum and average localization error for all $\mathbf{u}$'s denoted by $\Delta_{\max}$ and $\overline{\Delta}$. Moreover, because we have assumed that the feature is a random variable, the localization error $\Delta(\mathbf{u})$ can also be seen as  random variable. 

\begin{definition} [Achievable localization error] The maximum localization error of $\delta$ is achievable with probability $1-\epsilon$ if there is a fingerprinting algorithm such that :
$$
\Pr(\Delta_{\max}> \delta)\leq \epsilon.
$$
It is possible to similarly define achievability notion of the average localization error.
\end{definition}

It is important  to find all possible $(\delta,\epsilon)$ pairs. In the next section, it is shown that $\delta$ and $\epsilon$ can be made arbitrarily small by increasing training points and number of measurements. Although theoretically encouraging, the number of training locations and the number of measurements are essentially limited and extending the infrastructure is costly.  Under these constraints,  different trade-offs emerge between achievable error and different fingerprinting parameters. What is the minimum number of training locations $|\Lambda|$ that are required for achieving a certain localization error? How does the localization error scale with the number of measurements? How does it scale with the number of anchors? Some of these questions are addressed in this paper.

\vspace{-1mm}
\section{Theoretical Limits of Fingerprinting Algorithms Performance}
\label{limits}
In this section, a general fingerprinting algorithm is considered and bounds on its performance are investigated. The conditional distribution $\mathbb{P}_{X|\mathbf{u}}$ is a general distribution subject to the condition that the fingerprints are i.i.d., i.e., independent and identically distributed . For the moment, it is assumed that fingerprints are spatially robust and have minimum requirements for localization.

\subsection{Fingerprinting with Known Conditional Probability}
 
 As a first step, consider the case where a target node is located either at $\mathbf{u}_1$ or at $\mb u_2$. Moreover suppose that the conditional distributions $\mathbb{P}_{X|\mathbf{u}}$ are known. Two kinds of errors can occur. The first error occurs when the node is located at $\mb u_1$ but the reported location is $\mb u_2$. The second error occurs when the node is located at $\mb u_2$ but the reported location is $\mb u_1$. The probabilities of errors are defined as:	
\begin{align}
 \alpha(\mathbf{u}_1,\mathbf{u}_2)&=\mathbb{P}(g(\mb X)=\mathbf{u}_2 | \text{ Target node is at  } \mathbf{u}_1)\\
\beta(\mathbf{u}_1,\mathbf{u}_2)&=\mathbb{P}(g(\mb X)=\mathbf{u}_1 | \text{ Target node is at  } \mathbf{u}_2).
\end{align}
The goal is to minimize probabilities of both errors simultaneously. The problem of localization boils down to a decision between two probability distributions $\prob_{\mb X|\mb u_1}$ and $\prob_{\mb X|\mb u_2}$ based on the observation $\mb X$ and subject to constraints on $\alpha(\mathbf{u}_1,\mathbf{u}_2)$ and $\beta(\mathbf{u}_1,\mathbf{u}_2)$. 
The problem is indeed statistical hypothesis testing problem \cite{lehmann_testing_2006} where the goal is to decide between $\mathbb{P}_{X|\mathbf{u}_1}$ and $\mathbb{P}_{X|\mathbf{u}_2}$ based on $n$ samples. In this context, the errors $ \alpha(\mathbf{u}_1,\mathbf{u}_2)$ and $\beta(\mathbf{u}_1,\mathbf{u}_2)$ are respectively missed detection and false alarm. This formulation regards a design problem in fingerprinting localization as a hypothesis testing problem and benefits from abundant research materials in that area. As a first step, the following theorem provides fundamental limits on the performance of fingerprinting algorithms.  

\begin{theorem}
Consider a fingerprinting algorithm with fingerprints $\mb X\in\mathcal{X}^n$ consisting of $n$ i.i.d. fingerprints from conditional distribution as $\mathbb{P}_{X|\mathbf{u}}$. 
If the distribution $\mathbb{P}_{X|\mathbf{u}}$ is known to the fingerprinting algorithm, then there is a pattern matching function $g$, such that for each pair of locations $(\mathbf{u}_1,\mathbf{u}_2)$ and for any $0<\epsilon<1$, 
\begin{equation}
 \alpha(\mathbf{u}_1,\mathbf{u}_2) \leq \epsilon
\label{eq:error1}
\end{equation}
\begin{equation}
 \lim_{n\to\infty}\frac{1}{n}\log\beta(\mathbf{u}_1,\mathbf{u}_2)= -D(\mathbb{P}_{X|\mathbf{u}_1}\| \mathbb{P}_{X|\mathbf{u}_2}),
\label{eq:error2}
\end{equation}
where $D(\mathbb{P}_{X|\mathbf{u}_1}\| \mathbb{P}_{X|\mathbf{u}_2})$ is Kullback-Leibler divergence defined as\footnote{The definition can be generalized to cases where $\mathbb{P}_{X|\mathbf{u}_1}(x)$ is a measure absolutely continuous with respect to $\mathbb{P}_{X|\mathbf{u}_2}(x)$.}:
$$
D(\mathbb{P}_{X|\mathbf{u}_1}\| \mathbb{P}_{X|\mathbf{u}_2})=\int_{\mathcal{X}}\mathbb{P}_{X|\mathbf{u}_1}( x)\log\frac{\mathbb{P}_{X|\mathbf{u}_1}( x)}{\mathbb{P}_{X|\mathbf{u}_2}( x)}\mathrm{d} x.
$$
\vspace{-3mm}
\label{thm:1}
\end{theorem}
\begin{IEEEproof}
See Appendix. 
\end{IEEEproof}

Theorem~\ref{thm:1} provides fundamental bounds on error probability if maximum localization error is intended. First it shows that with large number of measurements, it is possible to distinguish correctly any two points in the space. But more importantly, it characterizes the latency of localization algorithm. The theorem shows that if $D(\mathbb{P}_{X|\mathbf{u}_1}\| \mathbb{P}_{X|\mathbf{u}_2})$ is non-zero, i.e. strictly positive, then the probability of error asymptotically decreases exponentially with $D(\mathbb{P}_{X|\mathbf{u}_1}\| \mathbb{P}_{X|\mathbf{u}_2})$. 
It can also be shown that this is the best exponent one can get under the first constraint on $\alpha(\mathbf{u}_1,\mathbf{u}_2)$ \cite{lehmann_testing_2006}. 
If the \ac{KL}-divergence of two probability distributions is small, this means that the error decreases with smaller exponent and therefore more measurements are needed to guarantee a fast decrease in error. In other words, $n$ should be increased to compensate for small divergence and reduce the error. 

The hypothesis testing formulation of localization problem provides the possibility of considering other fingerprinting  localization problems as an equivalent hypothesis testing problem. For instance, consider the problem of optimal fingerprint construction. From \ac{HT} point of view, the optimal fingerprint is the sufficient statistics and the pattern matching function is the statistical test. It is known that Neyman-Pearson tests are optimal tests. They achieve the minimum possible value for one of the errors under fixed bound on another error. A Neyman-Pearson test is a test in which the normalized observed log-likelihood ratio is compared with a threshold $\gamma$. Based on the observations $\mb X$, the normalized log-likelihood ratio $T_n$ is defined as:
\begin{equation}
 T_n=\frac 1n\log \frac{\mathbb{P}_{\mb X|\mathbf{u}_1}}{\mathbb{P}_{\mb X|\mathbf{u}_2}}.
 \label{eq:NPtest}
\end{equation}
If $T_n$ is bigger than $\gamma$, $\mathbf{u}_1$ is announced as the location and otherwise, $\mb u_2$ is announced. For Neyman-Pearson test, there are lot of results characterizing the asymptotic and non-asymptotic behavior of the errors in~\eqref{eq:error1} and~\eqref{eq:error2}. 

Another example is the map-aware localization. Having additional information about the indoor map and context of localization amounts to knowing a priori the probability that a target node is present at each location. This problem is equivalent to finding the best Bayes probability of error defined as:
$$
P_n^{(e)}= \prob(\mathbf{u}_1) \alpha(\mathbf{u}_1,\mathbf{u}_2)+\prob(\mathbf{u}_2) \beta(\mathbf{u}_1,\mathbf{u}_2).
$$
Interestingly Neyman-Pearson test with $\gamma=0$ is the optimal test.
\begin{proposition}
If the probability that the target node is present at $\mb u_1$ is in $(0,1)$, then there is a map-aware fingerprinting algorithm such that:
\begin{equation}
\liminf_{n\to\infty}\frac{1}{n}\log P_n^{(e)}=-I_0(0),
\end{equation}
where $I_0(0)$ is called the Chernoff information of the probabilities $\mathbb{P}_{\mb X|\mathbf{u}_1}$ and $\mathbb{P}_{\mb X|\mathbf{u}_2}$ and is defined as:
$$
I_0(0)=-\log\inf_{t\in\mathbb R}\esp\left( \exp\left(t \log\frac{\mathbb{P}_{  X|\mathbf{u}_2}}{\mathbb{P}_{ X|\mathbf{u}_1}}\right )\right)
$$
and the expectation is with respect to ${\mathbb{P}_{\mb X|\mathbf{u}_1}}$.
\label{prop:MapAware}
\end{proposition}
\begin{proof}
The proof is exact replication of equivalent hypothesis testing problem and follows from large deviation theory analysis. It can be found in \cite{dembo_large_2010}.  
\end{proof}
The important insight is that map-aware localization can be done using log-likelihood ration test $T_n$ in \eqref{eq:NPtest} with the threshold zero. In this case, the error decreases exponentially with $n$ and with the Chernoff information $I_0(0)$. Many other similar and interesting conclusions can be made from the hypothesis testing analogy. For instance, if Neyman-Pearson test is used, the errors decrease exponentially with rates dependent on the choice of threshold $\gamma$, similar to Proposition \ref{prop:MapAware}.  

\subsection{Fingerprinting with Unknown Conditional Probability}

The errors in~\eqref{eq:error1} and~\eqref{eq:error2} are obtained under the strong assumption that the algorithm knows the conditional probability distribution. 
One point of using fingerprinting is exactly to avoid the complexity behind characterizing the relation between locations and signal features. 
Hence, it is interesting to understand what happens if the conditional probability distribution is not known. One option is to use the training phase to learn the probability distribution at each location. Suppose that fingerprints take their value on a finite space $\mathcal{X}$. The assumption of finiteness is natural since the measurements are usually quantized and scaled version of the signal feature. Suppose that in the training phase, $\mb X$ is observed at $\mathbf{u}$. The empirical distribution of $\mb X$ can be defined as:
$$
\mathbb{Q}_{\mb X|\mathbf{u}}(x)=\frac{1}{n}\sum_{i=1}^n \mathbf{1}(X_i=x),
$$
where $ \mathbf{1}(X_i=x)$ is equal to one if $X_i=x$ and zero otherwise. 
\begin{proposition}
 The empirical distribution $\mathbb{Q}_{\mb X|\mathbf{u}}$ converges point-wise to the true distribution $\mathbb{P}_{\mb X|\mathbf{u}}$ exponentially fast with the number of measurements:
\begin{align}
 \Pr(\left|\mathbb{Q}_{\mb X|\mathbf{u}}(x)-\mathbb{P}_{\mb X|\mathbf{u}}(x)\right| \geq a)\leq 2e^{-2na^2}.
 \label{ineq:Empir}
\end{align}
\end{proposition}
\begin{IEEEproof}
$\mathbb{Q}_{\mb X|\mathbf{u}}(x)$ is sum of $n$ independent random variables $\mathbf{1}(X_i=x)$ with support in $\{0,1\}$. Using Hoeffding's  inequality \cite{hoeffding_probability_1963}, the above inequality is obtained. 
\end{IEEEproof}
Therefore for a finite fingerprinting space $\mathcal{X}$, it is guaranteed  with large $n$  that the conditional probability distribution $\mathbb{P}_{\mb X|\mathbf{u}}$ is recovered during the training phase. Note that the decrease exponent is $a^2$. It is possible to derive the exact exponent using large deviation theory. To this end, the total variation distance is used as a measure for closeness of probabilities:
$$
d_{\mathrm{TV}}(\mathbb{Q}_{X|\mathbf{u}},\mathbb{P}_{X|\mathbf{u}})=\frac{1}{2}\sum_{x\in\mathcal{X}} |\mathbb{Q}_{X|\mathbf{u}}(x)-\mathbb{P}_{X|\mathbf{u}}(x)|. 
$$
\begin{proposition}
A fingerprinting localization algorithm can estimate $\mathbb{P}_{X|\mathbf{u}}$ by the empirical distribution $\mathbb{Q}_{X|\mathbf{u}}$ such that:
\small
\begin{equation}
  \lim_{n\to\infty}\frac{1}{n}\log\mathbb{P}(d_{\mathrm{TV}}(\mathbb{Q}_{X|\mathbf{u}},\mathbb{P}_{X|\mathbf{u}}) > a)=- \inf_{\mathbb{P}\in B(a,\mathcal{X})} D(\mathbb{P}||\mathbb{P}_{X|\mathbf{u}})
\label{eq:errorsanov}
\end{equation}
\normalsize
where
$$
B(a,\mathcal{X})=\{\mathbb{P}: d_{\mathrm{TV}}(\mathbb{P},\mathbb{P}_{X|\mathbf{u}}) > a \}.
$$
\vspace{-4mm}
\label{prop:1}
\end{proposition}
\begin{IEEEproof}
The proposition is a direct consequence of Sanov's theorem~\cite{dembo_large_2010}. 
{Note that a similar analysis can be conducted for a general feature space $\mathcal{X}$, not necessarily finite. We avoid the technicalities here as the main conclusions remain similar. For the general analysis refer to~\cite[Chapter 6.]{dembo_large_2010}. }
\end{IEEEproof}

The previous proposition guarantees that the probability distribution can be estimated well during the training phase using $n$ number of measurements. Moreover the empirical distribution $\mathbb{Q}_{X|\mathbf{v}}$ can act as a fingerprint itself. 

The discussions in this section show that a fingerprinting algorithm does not need to know a priori the probability distribution relating a feature to a location. 
In general, for a spatially robust fingerprint,  there is a fingerprinting algorithm that can detect all locations without knowing the conditional probability distributions from enough large number of measurements. The error probabilities decrease exponentially with the number of measurements. 
\subsection{Fingerprinting with Training Grids}
So far, a two-point localization scenario has been considered. This can be easily extended to a localization scenario with finite candidate locations. Basically \ac{FPS} can reliably and accurately localize a target node within finite possible locations using enough large number of measurements. When possible locations are uncountable, the localization space is divided into multiple regions and the equivalent \ac{HT} problem aims at finding the region in which the target node is located. These regions are indeed related to the training grid. 
No assumption was made regarding the training database in previous part. Indeed, in order to have perfectly  accurate fingerprinting, training fingerprints are required for every points which is practically impossible. Suppose that training locations are limited to those in $\Lambda$, which is a finite set. The fingerprinting algorithm announces a training point as the estimated location. If a user is located at $\mb u$, then the question is how well one can localize the user by comparison to the training database. Following the discussion in Remark \ref{rem:Voronoi}, the estimated location is not necessarily the geometrically closest training point to the location $\mb u$ . Based on the specific chosen fingerprint and the pattern matching function, the modified Voronoi region of points, denoted by $\hat{\mathcal{V}}_{\mb v}$, is used. Localization in this scenario is equivalent to finding the modified Voronoi region $\hat{\mathcal{V}}_{\mathbf{v}}$ that contains the target node and  thus it determines the closest training location to the node. For this scenario, the following errors can be defined:
$$
\alpha(\hat{\mathcal{V}}_{\mathbf{v}_1},\hat{\mathcal{V}}_{\mathbf{v}_2})=\mathbb{P}(g(\mb X)=\mathbf{v}_2 | \text{ Target node is inside  } \hat{\mathcal{V}}_{\mathbf{v}_1}).
$$
$$
\beta(\hat{\mathcal{V}}_{\mathbf{v}_1},\hat{\mathcal{V}}_{\mathbf{v}_2})=\mathbb{P}(g(\mb X)=\mathbf{v}_1 | \text{ Target node is inside  } \hat{\mathcal{V}}_{\mathbf{v}_2}).
$$
Following theorem provides the fundamental limit of fingerprinting localization using  finite number of training location.

\begin{theorem}
Consider a fingerprinting algorithm with training locations in $\Lambda$. For the fingerprint $\mb X$, there is a pattern matching function $g$, based on the empirical distribution, such that for each $\epsilon>0$ and enough large $n$, $\alpha(\hat{\mathcal{V}}_{\mathbf{v}_1},\hat{\mathcal{V}}_{\mathbf{v}_2})\leq \epsilon$ and for $\mb u_2\notin  \hat{\mathcal{V}}_{\mathbf{v}_1}$, we have:

\begin{equation}
 \lim_{n\to\infty}\frac{1}{n}\log\beta(\hat{\mathcal{V}}_{\mathbf{v}_1},\hat{\mathcal{V}}_{\mathbf{v}_2})= -D(\mathbb{P}_{X|\hat{\mathcal{V}}_{\mathbf{v}_1}}\| \mathbb{P}_{X|{\mathbf{u}_2}}),
\label{eq:error4}
\end{equation}
where :
$$
D(\mathbb{P}_{X|\hat{\mathcal{V}}_{\mathbf{v}_1}}\| \mathbb{P}_{X|{\mathbf{u}_2}})=\inf_{\mb u \in \hat{\mathcal{V}}_{\mathbf{v}_1}} D(\mathbb{P}_{X|\mb u}\| \mathbb{P}_{X|{\mathbf{u}_2}}).
$$
\label{thm:1.1}
\vspace{-2mm}
\end{theorem}
\begin{IEEEproof} See Appendix.
\end{IEEEproof}

Theorem~\ref{thm:1.1} guarantees that fingerprinting algorithms can successfully find the region in which the target node is located. The region is the modified Voronoi region. 
Those points that are closer to the boundary of the region have the smallest $D(\mathbb{P}_{X|\hat{\mathcal{V}}_{\mathbf{v}_1}}\| \mathbb{P}_{X|{\mathbf{u}_2}})$. Intuitively this means that more measurements are needed for their localization. From Theorem \ref{thm:1.1}, the importance of spatial robustness of fingerprints becomes more clear. Without robustness, the modified Voronoi region can be so different from the original Voronoi regions that the closest fingerprint in the training database corresponds to a training point that is very far from the target node and hence leading to huge error.

\subsection{On  Training Grid Selection}
\label{sec:training}
One way to design the training grid is to consider its covering radius. The notion of covering radius is borrowed from lattice literature~\cite{Erez2005}. The covering radius of a region $\mathcal{V}_{\mathbf{v}}$ is defined as follows:
$$
r_{\text{cov}}(\mathcal{V}_{\mathbf{v}})=\inf\{r: \mathcal{V}_{\mathbf{v}}\subseteq\mathbf{v}+r\mathcal{B}\}
$$
where $\mathcal{B}$ is unit-radius ball in space. The covering radius is indeed the smallest ball centering at $\mathbf{v}$ and covering the region. 
The covering radius will provide the maximum localization error for the algorithm of  Theorem \ref{thm:1.1}.

\begin{corollary}
 Consider a fingerprinting algorithm with training locations in $\Lambda$ and modified Voronoi regions $\hat{\mathcal{V}}_{\mathbf{v}}$. The maximum error $ \Delta_{\max}$ is equal to the maximum of  all covering radius for each modified Voronoi regions:
 $$
 \Delta_{\max}=\max_{\mb v\in\Lambda}r_{\text{cov}}(\hat{\mathcal{V}}_{\mathbf{v}}).
 $$
\end{corollary}

Localization error depends therefore on the shape of Voronoi regions and modified Voronoi regions. The choice of training locations affects  the covering radius and it is desirable in general to have a larger covering radius for a fixed number of points. For example, the hexagonal lattice with the minimum distance between locations $d$ has the covering radius of $\frac{d}{\sqrt{3}}$.  The square lattice with the same minimum distance $d$ has larger covering radius $\frac{d}{\sqrt{2}}$. 
This means that the hexagonal grid is better that square lattice with respect to its coverage radius while keeping same minimum distance between training locations. 

Robustness of fingerprints guarantees that the covering radius of modified Voronoi regions does not change that much from the original Voronoi region. Suppose that the granularity of a training grid is increased. Increasing the granularity leads to smaller modified Voronoi regions. This means that the maximum value of $D(\mathbb{P}_{X|\hat{\mathcal{V}}_{\mathbf{v}_1}}\| \mathbb{P}_{X|{\mathbf{u}_2}})$ is decreased and therefore intuitively more number of measurements is needed to achieve the same accuracy. 
In other words, using finer grids is sometimes beyond need when the number of measurements is fixed. 

\vspace{-1mm}
\section{RSS-based Fingerprinting Algorithms}
\label{rss_fingerprinting}

So far, the feature used for fingerprinting has not been explicitly specified and it is modeled as a general random variable.  In this section, the introduced framework is used for a specific fingerprinting scenario. The most common choice for the signal feature is \ac{RSS} obtained from multiple anchors. The anchors are $K$ \acp{AP} belonging to a wireless network which is also used for fingerprinting localization.  They regularly transmit their signals and it is assumed that they are not interfering with each other. 
The interference avoidance can be obtained by employing proper \ac{MAC} mechanisms, although we do not consider any specific \ac{MAC} mechanisms in this work. 
We assume that the anchors are placed in a 2-dimensional Euclidean space.  The anchor $i$ is placed at the location $\mathbf{w}_i\in\mathbb{R}^2$ and transmits with the transmission power $P^{(i)}_{T}$. 
The transmitted signal in time is presented as $x^{(i)}(t)$ and \ac{RSS} at the location ${\mathbf{u}}$ as $y^{(i)}(\mathbf{u},t)$.

\subsection{Channel Model}
In order to explicitly express the dependence of \ac{RSS} from an anchor on the distance, the relation between $x^{(i)}(t)$ and $y^{(i)}(\mathbf{u},t)$ is shown as follows  \cite{tse_fundamentals_2005}:
$$
y^{(i)}(\mathbf{u},t)=\sum a^{(i)}_j(t)x^{(i)}(t-\tau^{(i)}_j(t)) +z^{(i)}(t)
$$
where $a^{(i)}_j(t)$ and $\tau^{(i)}_j(t)$ are the channel gain and delay of $j$'th multi-path component. The variation of these parameters determine the type of environment, namely indoor and outdoor, fast fading or slowly fading, frequency selective or frequency flat. In this work, the statistical channel model is considered where $a^{(i)}_j[n]=a^{(i)}_j(nT)$ is Rayleigh distributed ($T$ is sampling period). First consider only one dominant channel tap, which is $a^{(i)}_1[n]$ and $a^{(i)}_j[n]=0$ for $j>1$.  $a^{(i)}_1[n]$'s are correlated for different $n$ and their variation depends on the coherence time of the channel. They also incorporate the path loss. Two main scenarios are considered in this work. 

First it is assumed that the received power is calculated over long period and therefore the channel variations are averaged out. In this case, $a^{(i)}_1[n]$ is assumed to be $\frac{ 1}{\|{\mathbf{w}}_i-\mathbf{u}\|^{\alpha_i/2}}$. One example is WiFi RSSI value which varies only slightly in time due to quantization noise. The fingerprint $X^{(i)}_{\mathbf{u}}$ is defined equal to the received power and it is obtained as:
\begin{equation}
X^{(i)}_{\mathbf{u}}=P^{(i)}(\mathbf{u})=\frac{P^{(i)}_T}{\|{\mathbf{w}}_i-\mathbf{u}\|^{\alpha}}+N+\mb N_i,
\label{eq:QNoiseModel}
\end{equation}
where  $\alpha$ is the path loss exponent, $N$ is the additive noise power, $\mb N_i$ is a Gaussian random variable of variance $N_i$ to account for small changes in \ac{RSS} values due to quantization noise. The fingerprint at the point ${\mathbf{u}}$ is then $\mb X_{\mb u}=(X^{(1)}_{\mathbf{u}},\dots,X^{(K)}_{\mathbf{u}}).$

In second scenario it is assumed that instantaneous received power is considered where the channel coefficient changes in each round of \ac{RSS} calculation. This means that $a^{(i)}_1[n]$ is assumed to be $\frac{ \sqrt{\mb H^{(i)}}}{\|{\mathbf{w}}_i-\mathbf{u}\|^{\alpha_i/2}}$ where $\mb H^{(i)}$ is a random variable representing multipath fading and shadowing effect. No additional quantization noise is assumed in this case. 
Fingerprint is chosen as  \ac{RSS} value and it is calculated as follows:
\begin{equation}
 X^{(i)}_{\mathbf{u}}=P^{(i)}(\mathbf{u})=\frac{ \mb H^{(i)} P^{(i)}_{T}}{\|{\mathbf{w}}_i-\mathbf{u}\|^{\alpha_i}}+N.
\label{eq:FadingModel}
 \end{equation}
A database is created by measuring \ac{RSS} at the training locations ${\mathbf{v}}\in\Lambda$.  The measurements are used for localization of a target node located at the location $\mathbf{u}\in\mathcal{R}$. At each training location $\mathbf{v}$, \ac{RSS} is measured from each anchor $i$ and the measurement is repeated for number of times.

In the following, these two scenarios are studied using the framework developed above. Based on the results of previous section, the KL-divergence  $D(\mathbb{P}_{\mb X|\mathbf{u}_1}\| \mathbb{P}_{\mb X|\mathbf{u}_2})$ is adopted as the main metric of interest. This metric is inversely related with the latency of localization. Bigger KL-divergence indicates fewer measurements required for localization. 
\subsection{Noisy \ac{RSS} Fingerprinting}

\begin{figure*}[!th]
\vspace{-4mm}
\begin{minipage}{0.67\textwidth}
\centering
\begin{subfigure}{0.48\textwidth}
\centering
\includegraphics[width=\columnwidth]{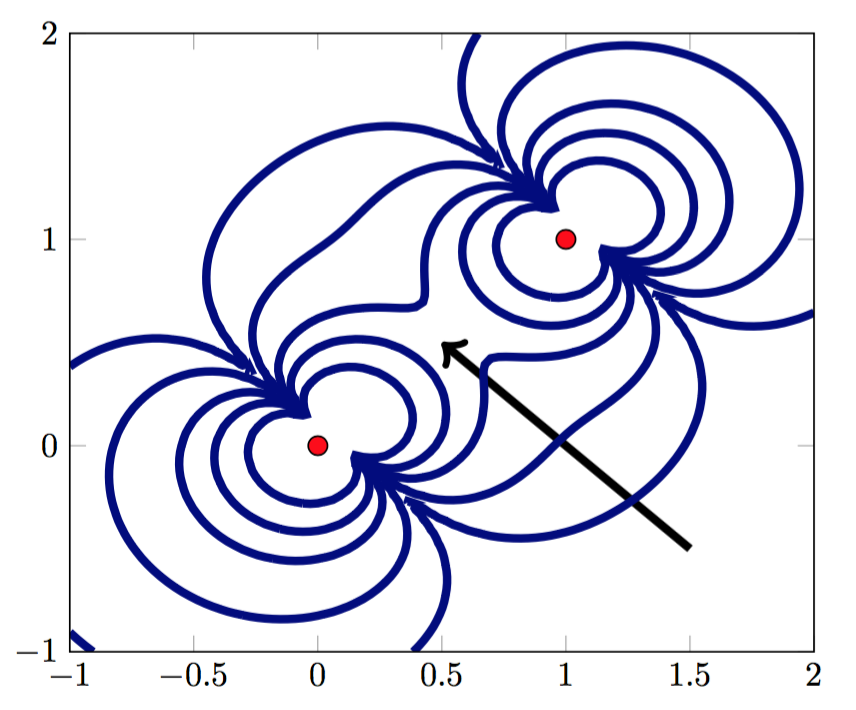}
\caption{Two APs level curves}
\end{subfigure} 
\begin{subfigure}{0.48\textwidth}
\centering
\includegraphics[width=\columnwidth]{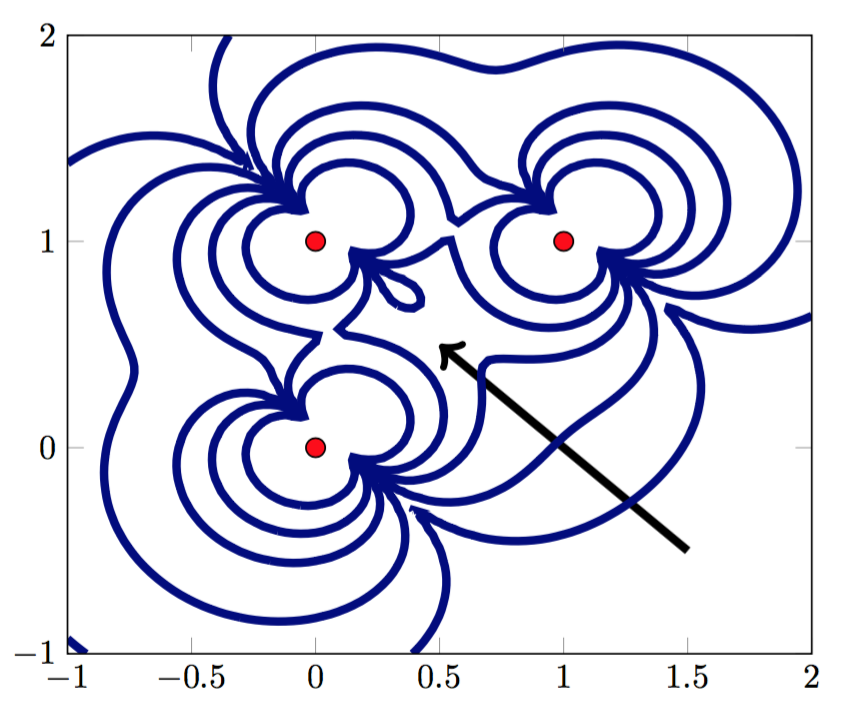}
\caption{Three APs level curves}
\end{subfigure}
\vspace{-2mm}
\caption{Level curves of \ac{KL}-divergence (localization latency)}
\label{fig:APlevels}
\vspace{-3mm}
\end{minipage}
\begin{minipage}{0.32\textwidth}
\vspace{1mm}
\includegraphics[width=\columnwidth]{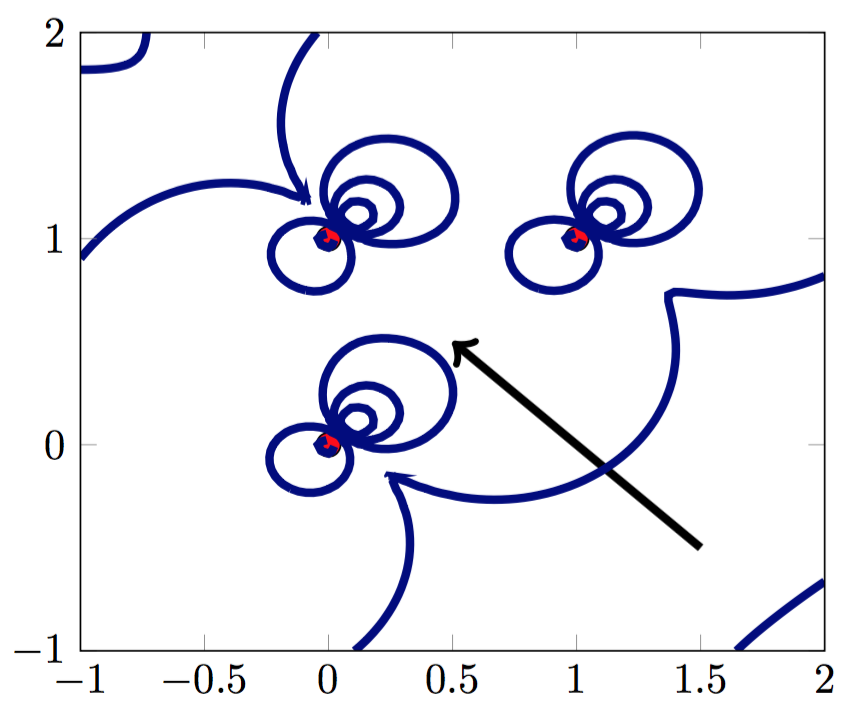}
\vspace{-7mm}
\captionsetup{justification=centering, margin=0.2cm}
\caption{Level curves of \ac{KL}-divergence (localization latency) under fading}
\label{fig:APlevelsfading}
\end{minipage}
\end{figure*}

For the choice of fingerprints \eqref{eq:QNoiseModel}, the metric $D(\mathbb{P}_{\mb X|\mathbf{u}_1}\| \mathbb{P}_{\mb X|\mathbf{u}_2})$ is \ac{KL} divergence of two normal random variable and it is evaluated as:
\small
\begin{equation}
D(\mathbb{P}_{\mb X|\mathbf{u}_1}\| \mathbb{P}_{\mb X|\mathbf{u}_2})=\sum_j\frac{(P^{(j)}_T)^2}{2N_j}\left(\frac{1}{\|{\mathbf{w}}_j-\mathbf{u}_1\|^{\alpha}}-\frac{1}{\|{\mathbf{w}}_j-\mathbf{u}_2\|^{\alpha}}\right)^2.
\label{eq:NoisyFPS}
\end{equation}
\normalsize
Note that if $D(\mathbb{P}_{\mb X|\mathbf{u}_1}\| \mathbb{P}_{\mb X|\mathbf{u}_2})$ is zero, then $\|{\mathbf{w}}_j-\mathbf{u}_2\|=\|{\mathbf{w}}_j-\mathbf{u}_1\|$ for all anchors. 
To guarantee uniqueness, the anchors should be placed in a way that if $\|{\mathbf{w}}_j-\mathbf{u}_2\|=\|{\mathbf{w}}_j-\mathbf{u}_1\|$ for all anchors, then $\mb u_1=\mb u_2$. This can be achieved simply by choosing three non-collinear anchors in $\mathbb{R}^2$. 

Let's evaluate the robustness of this fingerprint. To this purpose, it is assumed that the localization is done over $\mathbb{R}^+$ with  an anchor placed at origin and the localization space inside the interval $[0,D]$. 
\begin{proposition}
Suppose that the localization space lies in $[0,D]$. For an anchor placed at origin, \ac{RSS}-fingerprints are spatially robust. In other words, if $D( \prob_{\mb X|\mb u_1}\| \prob_{\mb X|\mb u_2}) \leq L$ then
$$
 |{ u_1- u_2}|\leq  \frac{D^{\alpha+1}}{\alpha P_T} \sqrt{2N_1 L}
 $$
 \label{prop:stab}
 \end{proposition}
\begin{IEEEproof}
 See Appendix. 
\end{IEEEproof}
This result can be easily extended to localization in general Euclidean space, however even this simple consideration is illuminating. First, if the localization space, i.e., $D$ is large, then the arbitrarily far locations from the anchor can have very small $D( \prob_{\mb X|\mb u_1}\| \prob_{\mb X|\mb u_2})$ while $|{ u_1- u_2}|$ is arbitrarily large. Therefore \ac{RSS} fingerprints are not robust in general for an unbounded localization space. Another interpretation of this observation is that the new anchor should be placed as far as possible from the first anchor to increase the overall robustness of the fingerprint. The proposition \ref{prop:stab} indicates also how far one can measure \ac{RSS} of an anchor to guarantee a certain accuracy. Moreover the robustness is improved with the transmission power of the anchor and path loss exponent. 

The relation between $|{ u_1- u_2}|$ and  $D( \prob_{\mb X| u_1}\| \prob_{\mb X| u_2})$ can be used to provide bounds on the accuracy. For instance consider a typical setting for WiFi fingerprinting. Let the localization region be a line of 10 meter, i.e., $D=10m$, transmission power be $P=100 mW$, \ac{RSS} noise be around $N_1=10^{-3}$ and the granularity of devices permit separating  fingerprints by precision of $L=10^{-3}$. Those points that have $D( \prob_{\mb X| u_1}\| \prob_{\mb X| u_2})\leq L= 0.001$ are at most 5 cm apart, which means that the \ac{FPS} can achieve at best around 5 cm geometric error. 
This provides guidelines on how the training locations should be chosen to guarantee a certain accuracy.

In general, increasing number of anchors can improve KL-divergence and thereby reduce the number of required measurements for localization to achieve certain accuracy. This is subject to proper placement of anchor points. Note that the worst performance are obtained at the points with smallest KL-divergence and therefore new anchors should be placed in a way to increase the KL-divergence exactly for those points. To see this, suppose that the goal is to be able to distinguish all points with distance more than $d$. Consider all points $\mathbf{u}_2$ and $\mathbf{u}_1$ such that $\|\mathbf{u}_2-\mathbf{u}_1\|=d$. One way to study this problem is to look at level curves of the function: 
$$
\ell(\mb u,\mb e)=\sum_j\frac{(P^{(j)}_T)^2}{2N_j}\left(\frac{1}{\|{\mathbf{w}}_j-\mathbf{u}\|^{\alpha}}-\frac{1}{\|{\mathbf{w}}_j-\mathbf{u}-\mb e\|^{\alpha}}\right)^2,
$$
where $\mb e$ is an arbitrary vector of norm $d$. Figure \ref{fig:APlevels} shows these level curves for two and three anchors and $\mb e=(0.1,0.1)$ where the arrows show the direction of increase in \ac{KL}-divergence. Changing $\mb e=(0.1,0.1)$ rotates the plots according to the new chosen vector. This does not change significantly the shape of curves far from anchors unlike those close to anchors. If the number of measurements to achieve certain accuracy is an indicator of latency, these curves show the latency of localization for distinguishing $\mb u$ and $\mb u+\mb e$. 

It can be seen that the location pairs with the same $\|\mathbf{u}_2-\mathbf{u}_1\|$ that are closer to \acp{AP} have larger difference in their fingerprints. The larger difference between fingerprints provide more robust scheme to eventual localization errors. 

It can be seen that the small value of $\ell(\mb u,\mb e)$, which indicates bad localization performance, corresponds to an oval surrounding the anchors and particularly to points between the anchors.  This observation suggests that a new anchor should be placed on those curves containing the localization area. Moreover one option to find the position of new anchors is to consider Voronoi regions of current anchors and place new anchor on the intersection of Voronoi regions. In this way, \ac{KL}-divergence is increased by installing nearby anchors to compensate the effect of far anchors. 

The expression of KL-divergence is in general complex for evaluation. However, applying mean value theorem for several variables yields that:
\begin{align*}
D(\mathbb{P}_{\mb X|\mathbf{u}_1}\| \mathbb{P}_{\mb X|\mathbf{u}_2})\leq \sum_j\frac{\alpha^2  (P^{(j)}_T)^2}{2N_j}\left(\frac{\|\mb u_1-\mb u_2\|^2}{\|\overline{\mathbf{u}}_j-\mathbf{w}_j\|^{2\alpha+2}}\right)
\end{align*}
where $\overline{\mathbf{u}}_j$ is a location on the line between $\mathbf{u}_1$ and $\mathbf{u}_2$. For enough far points from all anchors, all $\overline{\mathbf{u}}_j$ can be approximated as equal and one can use the function $\tilde{\ell}(\mb u)=\displaystyle\sum_j\frac{1}{\|\mathbf{u}-\mathbf{w}_j\|^{2\alpha+2}}$ for evaluating the effect of anchor placements on localization performance. First, for enough small or enough large path loss exponent $\alpha$, the distance between two fingerprints is arbitrarily small. Moreover the level curves of $\tilde{\ell}(.)$ approximate well that of $\ell(.)$ for far points from anchors. $\tilde{\ell}$ provides a better understanding in general since it does not depend on $\mb e$.

\vspace{-1mm}
\subsection{Fingerprinting under Fading}
\label{sec:FPSfading}

In this part, it is assumed that the channel gains are random variables. 
The randomness can be due to shadowing or multi-path fading. 
For the anchor $i$, the \ac{RSS} value is modeled as: 
$$
X^{(i)}_{\mathbf{u}}=\frac{ \mb H^{(i)}  P^{(i)}_{T}}{\|{\mathbf{w}}_i-\mathbf{u}\|^{\alpha_i}}+N.
$$
Let us again assume a single anchor scenario. 
From Theorem~\ref{thm:1}, we look at the \ac{KL}-divergence of \ac{RSS} values for different locations. Considering standard Rayleigh distribution for fading, $\mb H^{(i)}$ is exponentially distributed and we have:
$$
D(\mathbb{P}_{\mb X|\mathbf{u}_1}\| \mathbb{P}_{\mb X|\mathbf{u}_2})=\sum_j \left(\log \frac{\|{\mathbf{w}}_j-\mathbf{u}_2\|^{\alpha}}{\|{\mathbf{w}}_j-\mathbf{u}_1\|^{\alpha}} +\frac{\|{\mathbf{w}}_j-\mathbf{u}_1\|^{\alpha}}{\|{\mathbf{w}}_j-\mathbf{u}_2\|^{\alpha}}-1\right).
$$
First, we show that the fading improves the localization performance by using KL-divergence metric. To compare this with noisy case, a standard inequality $x-1-\log x\geq\frac{1}{2}(x-1)^2$ for $x<1$ can be used to show that:
\begin{equation}
\begin{split}
\log \frac{\|{\mathbf{w}}_j-\mathbf{u}_2\|^{\alpha}}{\|{\mathbf{w}}_j-\mathbf{u}_1\|^{\alpha}} +\frac{\|{\mathbf{w}}_j-\mathbf{u}_1\|^{\alpha}}{\|{\mathbf{w}}_j-\mathbf{u}_2\|^{\alpha}}-1 \geq \frac {\|{\mathbf{w}}_j-\mathbf{u}_1\|^{2\alpha}}2\left(\frac{1}{\|{\mathbf{w}}_j-\mathbf{u}_1\|^{\alpha}}-\frac{1}{\|{\mathbf{w}}_j-\mathbf{u}_2\|^{\alpha}}\right)^2,
 \label{ineq:FadingVSNoise}
\end{split}
\end{equation}
where without loss of generality, we assumed that $\|{\mathbf{w}}_j-\mathbf{u}_1\|\leq \|{\mathbf{w}}_j-\mathbf{u}_2\|$. The right hand side of the  inequality in \eqref{ineq:FadingVSNoise} shows an additional product factor $\|{\mathbf{w}}_j-\mathbf{u}_1\|^{2\alpha}$ compared to \eqref{eq:NoisyFPS} ignoring the constant factors. The inequality in \eqref{ineq:FadingVSNoise} shows that for points far from the anchors ($\|{\mathbf{w}}_j-\mathbf{u}_1\|>1$), \ac{KL}-divergence is larger in fading environments and therefore localization performance is superior. The situation is reverse for points closer to the anchors, $\|{\mathbf{w}}_j-\mathbf{u}_1\|<1$. This can also be verified using level curves of $\ell(\mb u,\mb e)$ as defined above in Figure~\ref{fig:APlevelsfading}. In other words, shadowing and fading are expected to improve the worse case performance of fingerprinting algorithms intuitively because they create more variability in the \ac{RSS} pattern of different locations. Regarding anchor placement, level curves of $\ell(\mb u,\mb e)$ suggest that the same guideline regarding Voronoi-based anchor placement holds far fading case too. 
\vspace{-1mm}

\begin{figure*}[!th]
\vspace{-2mm}
\centering
\begin{subfigure}{0.32\textwidth}
\centering
\includegraphics[width=\columnwidth]{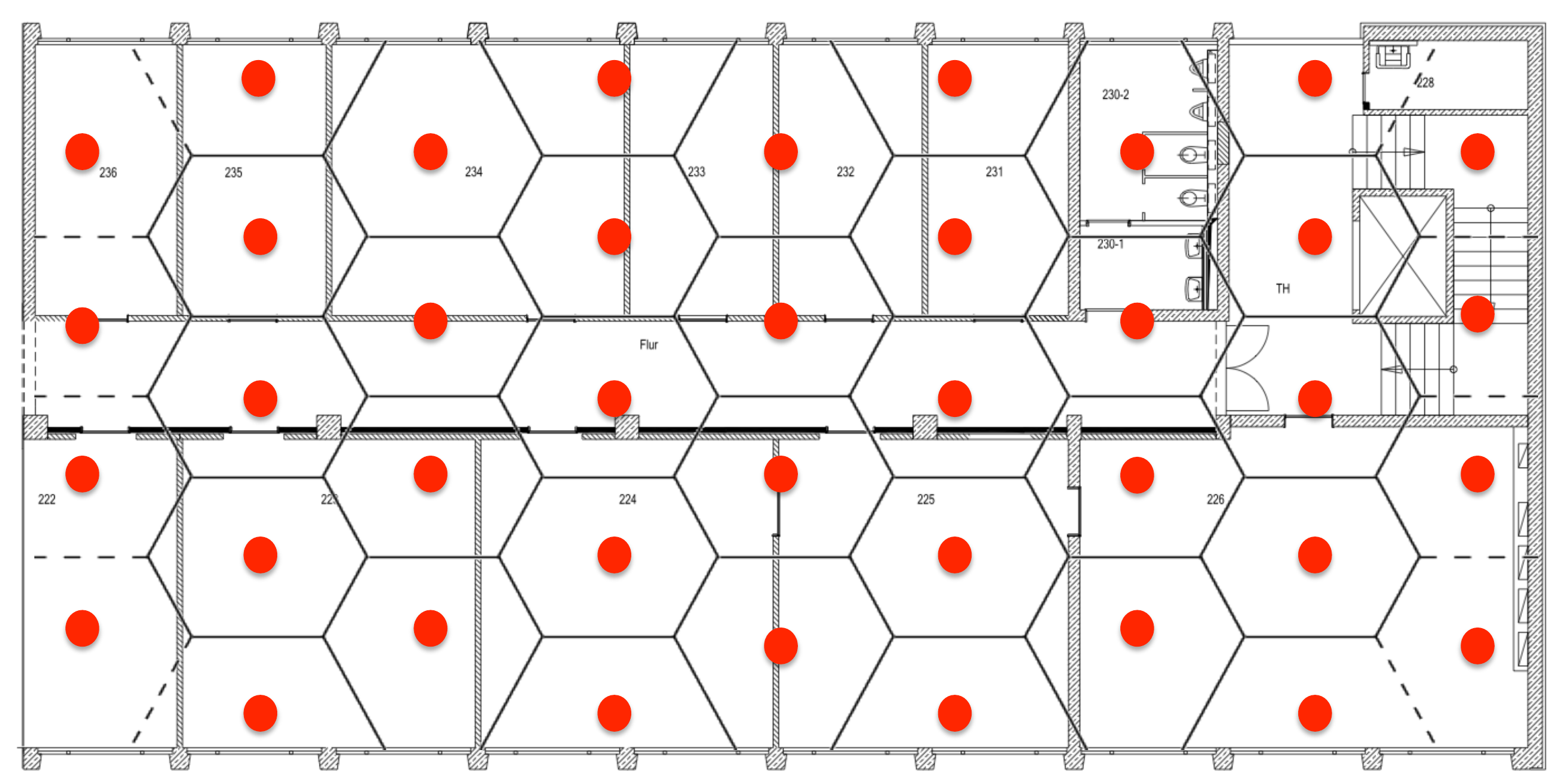}
\caption{Hexagonal training grid}
\end{subfigure} \hspace{1mm}
\begin{subfigure}{0.32\textwidth}
\centering
\includegraphics[width=\columnwidth]{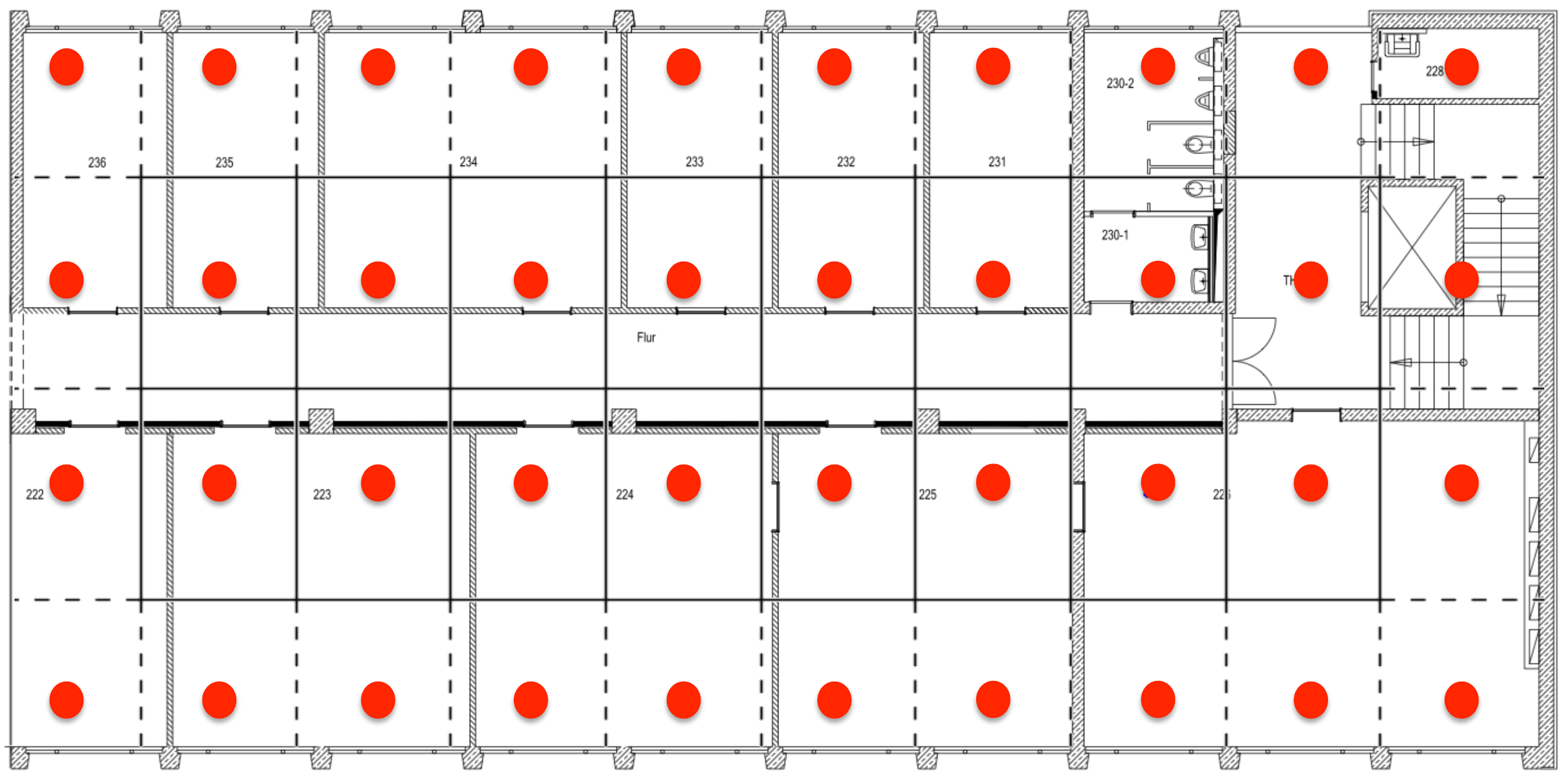}
\caption{Squared training grid}
\end{subfigure} \hspace{1mm}
\begin{subfigure}{0.32\textwidth}
\centering
\includegraphics[width=\columnwidth]{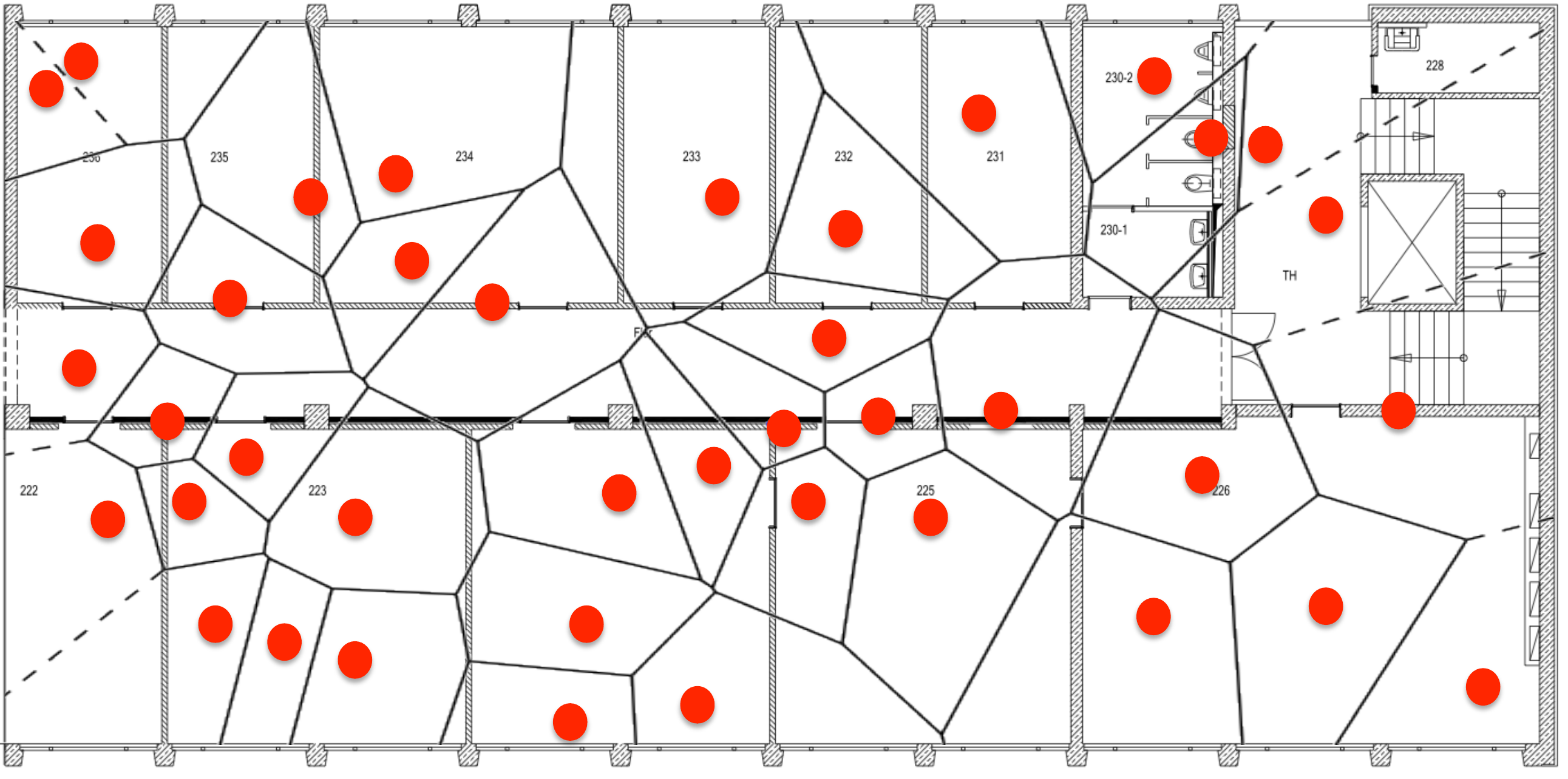}
\caption{Random training grid}
\end{subfigure}
\vspace{-2mm}
\caption{Selected training grids}
\label{fig:training_grid}
\vspace{-4mm}
\end{figure*}

\begin{figure*}[!ht]
\begin{minipage}{0.32\textwidth}
\vspace{2mm}
\includegraphics[width=\columnwidth]{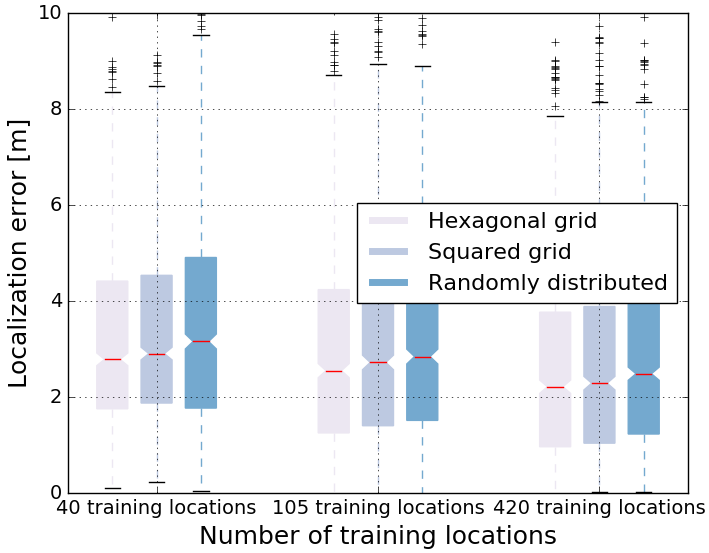}
\captionsetup{justification=centering}
\vspace{-5mm}
\caption{Localization error vs. grid selection \newline}\label{fig:training_points}
\end{minipage} \hfil
\begin{minipage}{0.31\textwidth}
\includegraphics[width=\columnwidth]{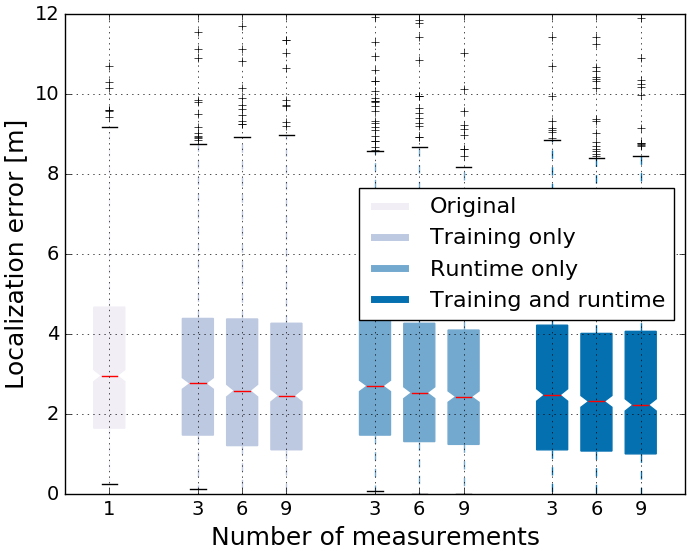}
\captionsetup{justification=centering}
\vspace{-5mm}
\caption{Localization error vs. number of measurements}\label{fig:number_meas}
\end{minipage} \hfil
\begin{minipage}{0.31\textwidth}
\includegraphics[width=\columnwidth]{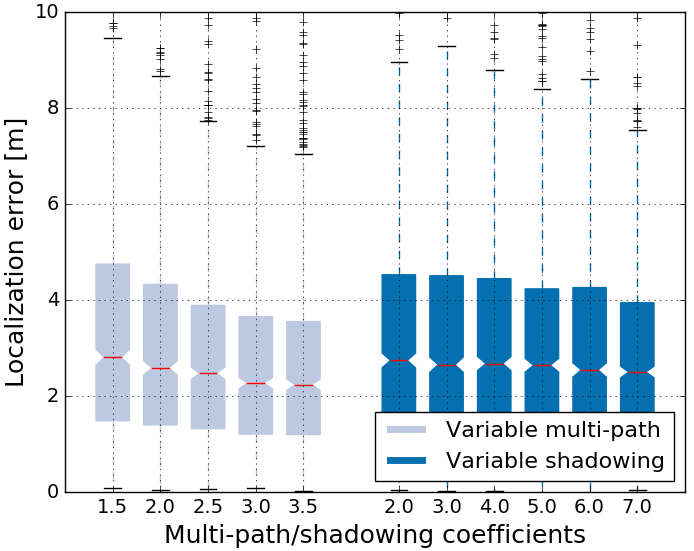}
\captionsetup{justification=centering}
\vspace{-5mm}
\caption{Localization error vs. shadowing and fading}\label{fig:shadowing_fading}
\end{minipage}
\vspace{-6mm}
\end{figure*}

\section{Evaluation}
\label{experiments}

In this section, the implications of the developed theory are examined. 
Through numerical simulation and experimental results, it is shown that the guidelines provided by the above framework are consistent in practical cases of interest. As a fingerprint at each location, the vector of average \ac{RSSI} values observed from different WiFi \acp{AP} is selected, which is a well-known and extensively used fingerprint selection method~\cite{Lemic14experimental_decomposition_of,lemic_quantile}. 
The similarity kernel is Euclidean distance between RSSI vectors, which is again a standard and often used method in fingerprinting~\cite{lemic15_open_challenge}.   

\subsection{Evaluation through Simulation}
In a simulation environment, a set of \acp{AP} related parameters, i.e. their locations and transmit powers are fixed. 
\ac{RSSI} values obtained from each \ac{AP} at a target node with an unknown location are modeled using the COST 231 multi-wall model for indoor radio propagation~\cite{borrelli2004channel}.
The applicability of the model has been demonstrated for localization purposes~\cite{caso15onthe} and the model has been previously used for localization related simulations (e.g.~\cite{lemic16towards,lemic16_enriched}).
The model accounts for the type and number of walls or obstacles in an environment, as well as for the locations of \acp{AP}.

In this model, the first attenuation effect is modeled by a well-known one-slope term that relates the received power to the distance. 
This term is tuned by the constant $l_0$, i.e. the path-loss at 1 m distance from the anchor and at the center frequency of 2.45 GHz, and the path-loss exponent $\gamma$. 
The second attenuation effect is modeled by a linear wall attenuation term. 
The number of walls in the direct path between a transmitter and a target node is counted and an attenuation contribution is assumed for each wall. 
The model outputs \ac{RSS} values from the defined \acp{AP} at a target node's location.  
A noise is then added to the derived \ac{RSS} values, drawn from a Gaussian distribution $\mathcal{N}(0,\sigma)$.
Gaussian noise is commonly introduced in indoor positioning related simulations to account for variations such as quantization error (e.g.~\cite{iwcmc}). 
For the simulation and later experimental examination environment, the {TWIST} testbed is selected~\cite{lemic_infrastructure}.
The {TWIST} testbed environment is an office building, with its outline given in Figure~\ref{fig:training}.  
In the model parameterization, measurements from the TWIST testbed were used and then a least-square fitting procedure is leveraged to minimize the cost function between the measured received power and the modeled one.

The input parameters of the model are the constant $lc$ related to the least square fitting procedure, the path-loss exponent $\gamma$, and the wall attenuation factor $l_w$. 
Additionally, a zero-mean Gaussian noise with standard deviation $\sigma$ has been added to the obtained \ac{RSS} values. 
If not explicitly stated otherwise, the parameters are chosen as  $lc=53.73$, $\gamma=1.64$, $l_w=4.51$, $\sigma=2$ for deriving our simulation results. 
The transmit power of each AP equals 20~dBm. 
For the majority of our simulation results, we defined a set of 4~\acp{AP}, with their locations indicated in Figure~\ref{fig:training}.
A target's node location has been selected randomly, its location has been estimated using the selected fingerprinting algorithm, and the localization error, i.e. an offset from the true location has been calculated.
The procedure has been repeated 10000 times and the results have been reported in a regular box-plot fashion.

First, the claims of the paper regarding the density of training locations and the selection of the training grid in Section~\ref{sec:training} are examined. 
The performance of the algorithm is evaluated for cases where the selected training grid is hexagonal, squared, and random, as shown in Figure~\ref{fig:training_grid}. 
Furthermore, the effect of the number of points in a training grid is evaluated. 
In addition to the depicted grids, each containing 40 training locations (3~m cell size for the squared grid), we also evaluated the case with 105 (2~m cell size) and 420 (1~m cell size) training locations. The results are depicted in Figure~\ref{fig:training_points}. 
As it can be seen in the figure, the usage of a hexagonal grid, i.e. the one with the largest covering radius, yields slightly better localization errors in contrast to a squared grid, while the random grid yields the worst performance. 
These results are in accordance with the developed theory that states that it is in general desirable to have a larger covering radius for a fixed number of training locations which favors the choice of hexagonal grids. Furthermore, the increase in the number of training locations improves the performance of fingerprinting algorithms. 
However, increasing the training locations by order of 10 improves only slightly the accuracy (Figure~\ref{fig:training_points}).
Both results are in accordance with the respective theory stating that an increase in the number of training locations improves the accuracy, but a training grid with very fine granularity is often not of practical benefit.

Second, the statement given in Theorem~\ref{thm:1} is evaluated concerning the number of observations at both training locations and in the runtime phase of fingerprinting. 
The training grid is fixed to a hexagonal one of 105 training locations. 
Furthermore, we increase the number of observations from 1 to 15 with a step of 5 in both training and runtime phase of fingerprinting. 
The results are presented in Figure~\ref{fig:number_meas}. 
As visible in the figure, in comparison to the basic case where only one measurement is taken in both training and runtime phase, an increase in the number of measurements generally reduces the localization errors.
Furthermore, the reduction of errors is higher in case one increases the number of measurements in both runtime and training phase, in contrast to increasing this number in one phase only. 
Both results are consistent with the respective theoretically derived conclusions stating that an increase in the number of measurements in both phases of fingerprinting optimally compensates for a small KL-divergence and in general reduces localization errors.

Third, we evaluate the statement given in Section~\ref{sec:FPSfading} stating that an increase in fading and shadowing propagation characteristics benefits the accuracy of fingerprinting. 
In the simulation model, multi-path fading is characterized by changing the path-loss exponent $\gamma$, while shadowing depends on the wall attenuation term $l_w$ . 
Therefore, to evaluate the multi-path fading effect on the performance, the path-loss exponent is increased from 1.5 to 3.5, while for characterizing the effect of shadowing we increase the wall attenuation term from 2 to 7~dBm. 
The achieved localization errors for such scenarios are given in Figure~\ref{fig:shadowing_fading}. 
The results demonstrate that an increase in both multi-fading and shadowing yields benefits for the performance of fingerprinting algorithms, which is in accordance to the derived theory.

Finally, we evaluate the statement given in Proposition~\ref{prop:stab} and in Section~\ref{rss_fingerprinting} where we consider the benefits of increasing the number of APs and properly placing them in a targeted environment. 
Starting from the basic scenario with 4~\acp{AP} (AP1,~AP2,~AP3,~AP4) depicted in Figure~\ref{fig:aps_locations}, two additional APs (AP5 and AP6 in Figure~\ref{fig:aps_locations}) are introduced in the environment, where their locations are selected either randomly or based on Voronoi vertices. 
Voronoi vertices-based selection places new APs at locations that are the farthest from the locations of existing APs. 
Furthermore, based on the locations of now 6 APs, additional 4 APs are introduced both randomly and, as
depicted in Figure~\ref{fig:aps_locations}, based on Voronoi vertices. 
The localization errors for such scenarios are depicted in Figure~\ref{fig:number_aps}. 
As visible in the figure, the increase in the number of APs generally improves the accuracy of fingerprinting. 
This is consistent with the developed theory stating that an increase in the number of APs improves KL-divergence and thereby reduces the number of required measurements for localization to achieve certain accuracy.
Secondly, in comparison to 5 different random sections both in case of 6 and 10 APs, Voronoi vertices-based selection generally achieves better performance. 
In other words, selection of APs in a way that they are the farthest from the locations of existing APs outperforms other selections.
This is consistent with the theoretically derived statement saying that the KL-divergence is increased by installing nearby APs to compensate the effect of far APs, which can be achieved by placing new APs on the intersection of Voronoi regions.

\begin{figure}[!th]
\vspace{-3mm}
\centering
\includegraphics[width=0.50\columnwidth]{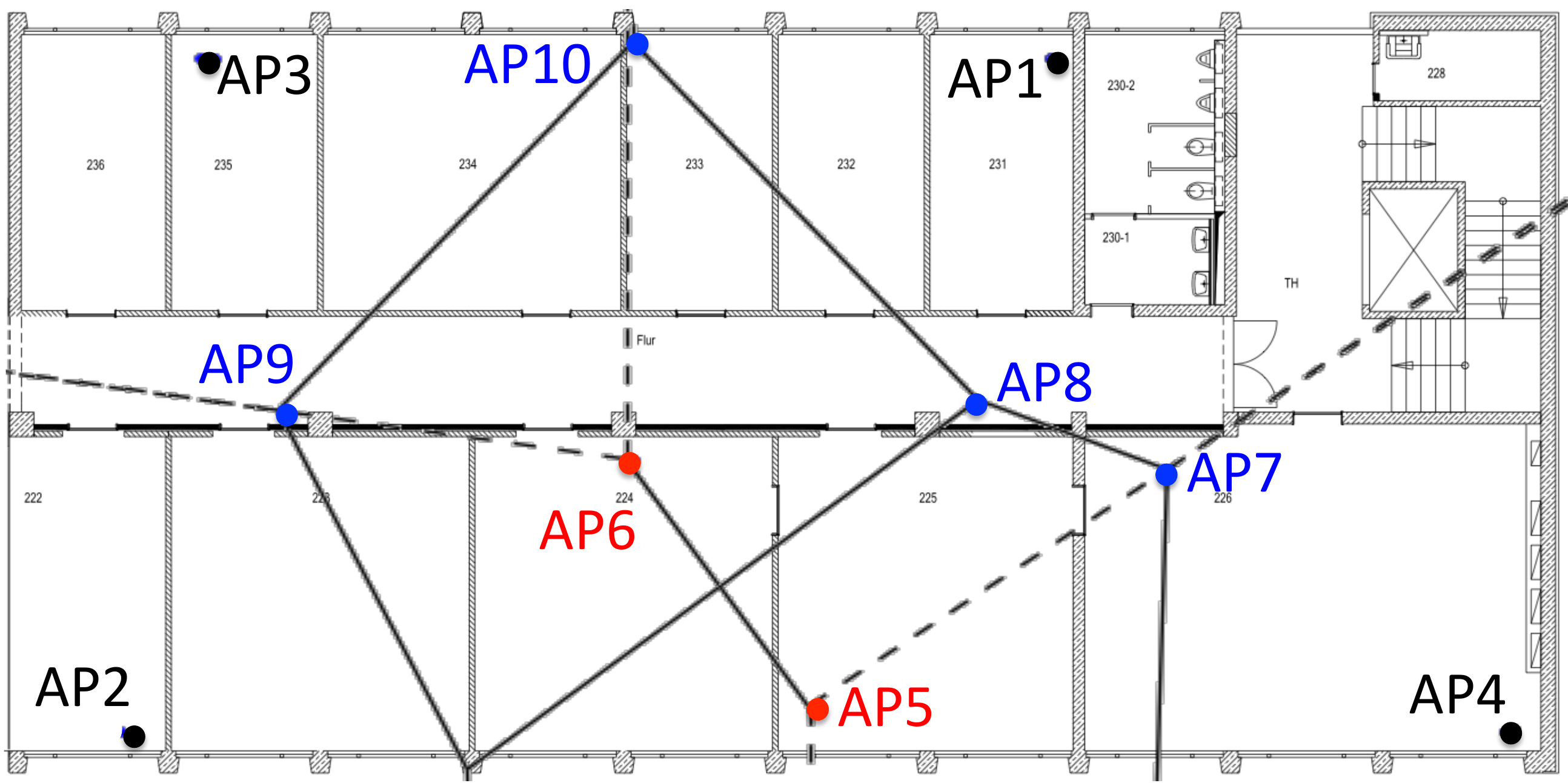}
\vspace{-1mm}
\caption{Voronoi vertices-based selection of APs' locations} 
\label{fig:aps_locations}
\vspace{-0mm}
\end{figure}

\begin{figure}[!th]
\centering
\includegraphics[width=0.5\columnwidth]{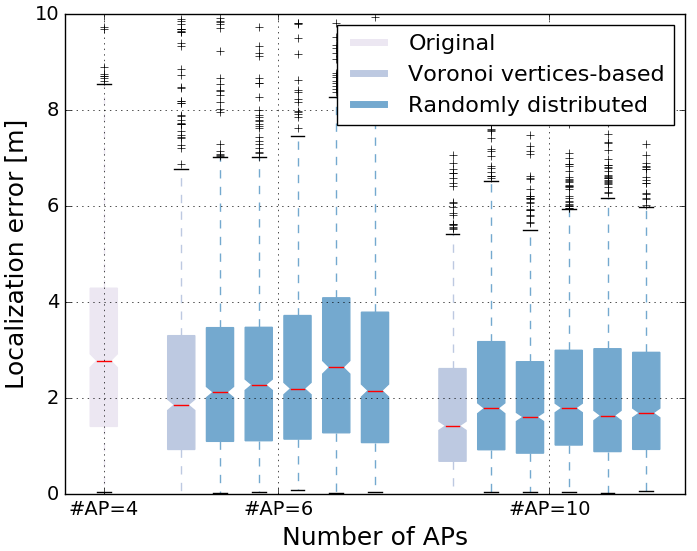}
\vspace{-1mm}
\caption{Localization error vs. number and locations of APs} 
\label{fig:number_aps}
\vspace{-3mm}
\end{figure}

\begin{figure}[!th]
\vspace{-1mm}
\centering
\begin{subfigure}{0.5\columnwidth}
\centering
\includegraphics[width=0.8\columnwidth]{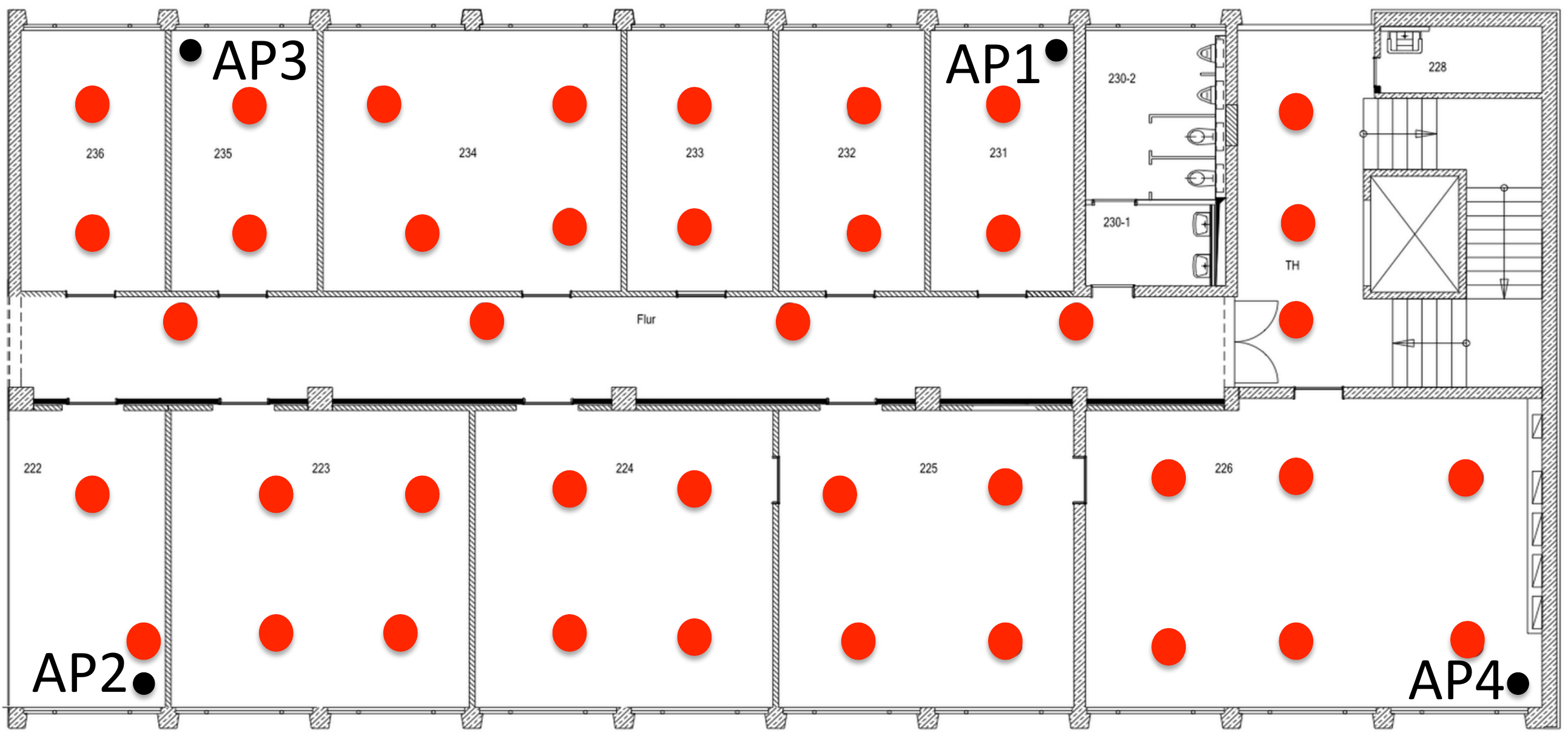}
\caption{Training locations}
\end{subfigure} \linebreak \vspace{2mm}
\begin{subfigure}{0.5\columnwidth}
\centering \vspace{2mm}
\includegraphics[width=0.8\columnwidth]{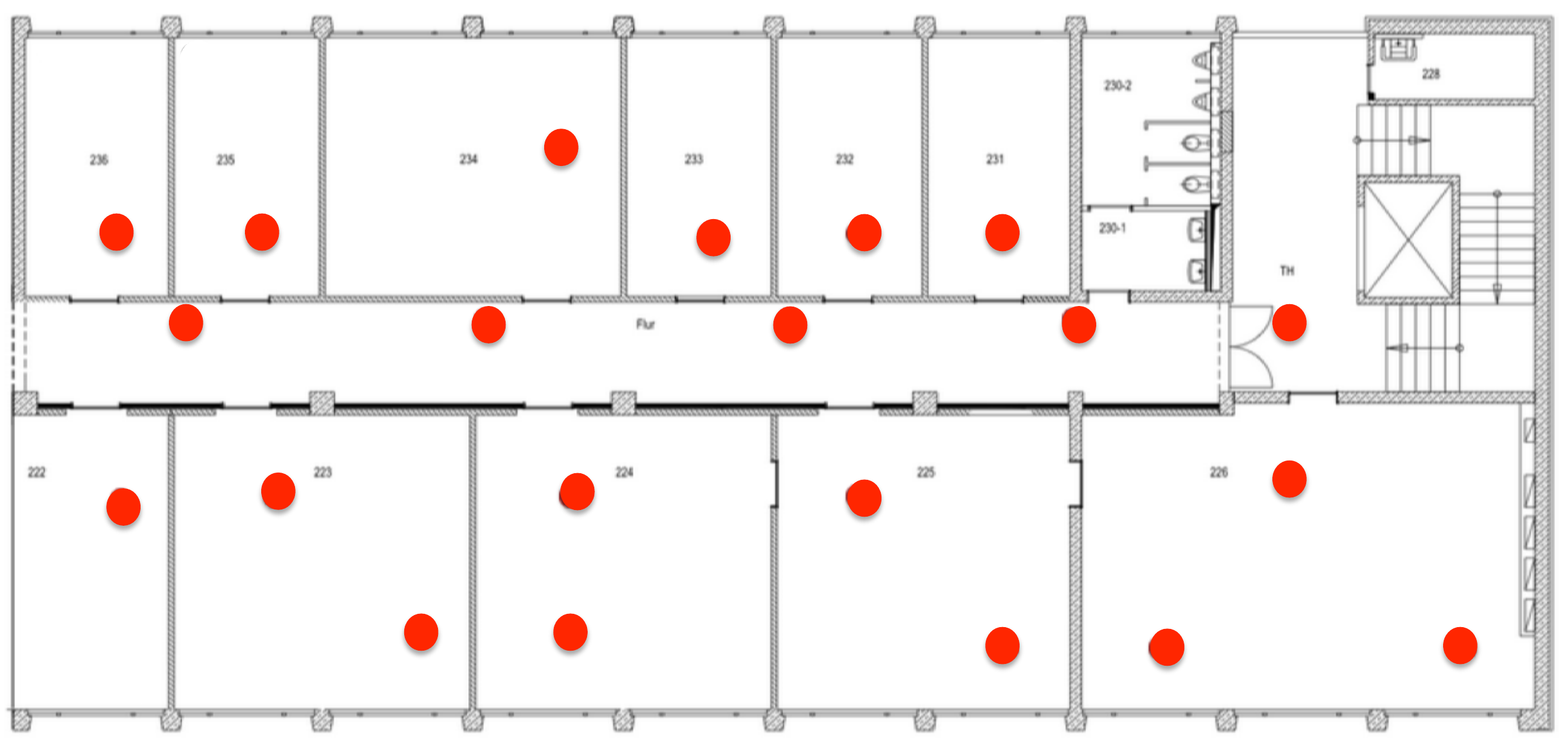}
\caption{Evaluation locations}
\end{subfigure}
\vspace{-2mm}
\caption{Training and evaluation locations}
\label{fig:training}
\vspace{-4mm}
\end{figure}

\begin{figure*}[!ht]
\centering
\begin{subfigure}{0.32\linewidth}
\centering
\includegraphics[width=\columnwidth]{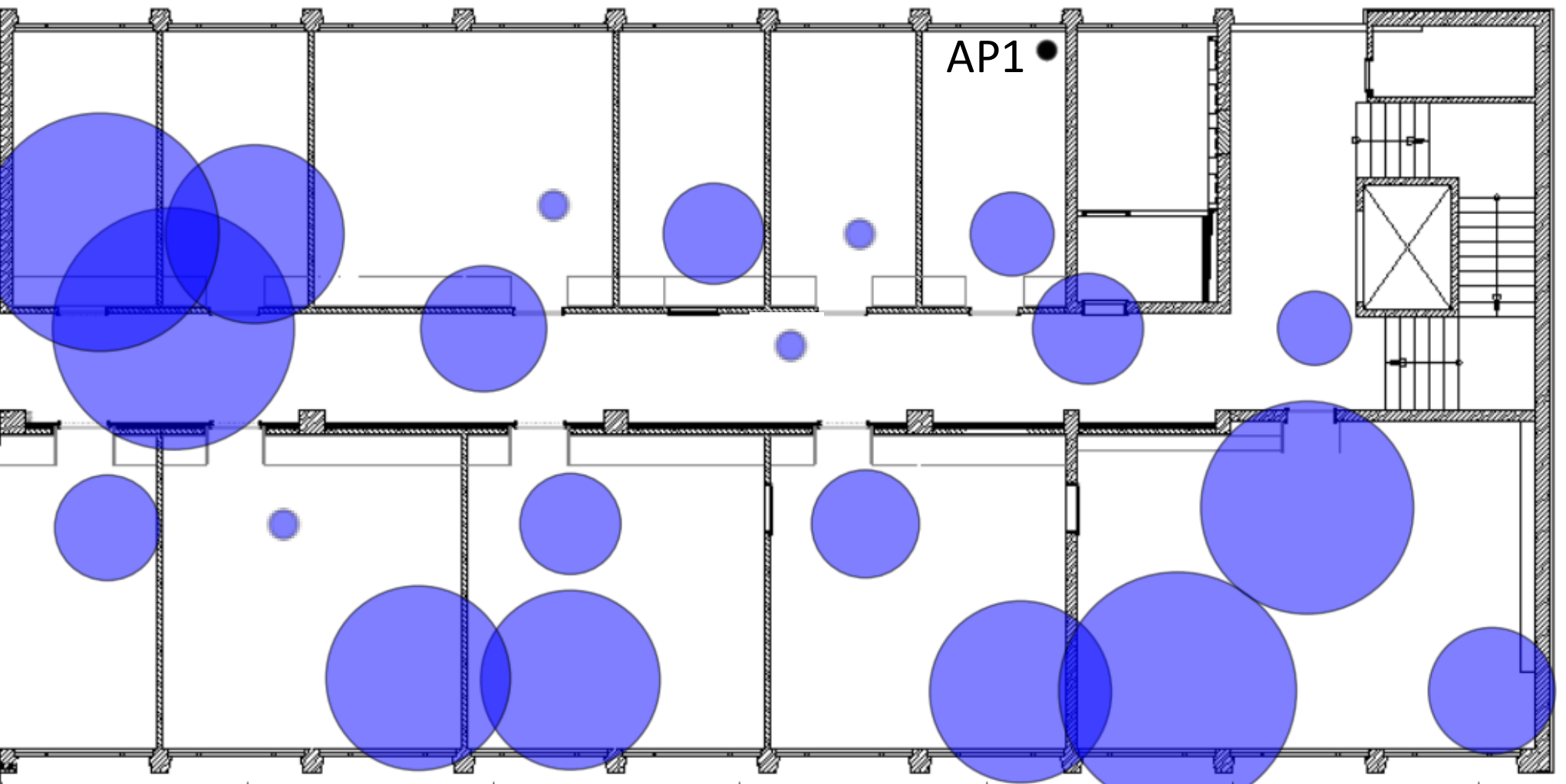}
\caption{Fingerprinting with AP1}
\label{fig:errordist1}
\end{subfigure} \hspace{1mm}
\begin{subfigure}{0.32\linewidth}
\centering
\includegraphics[width=\columnwidth]{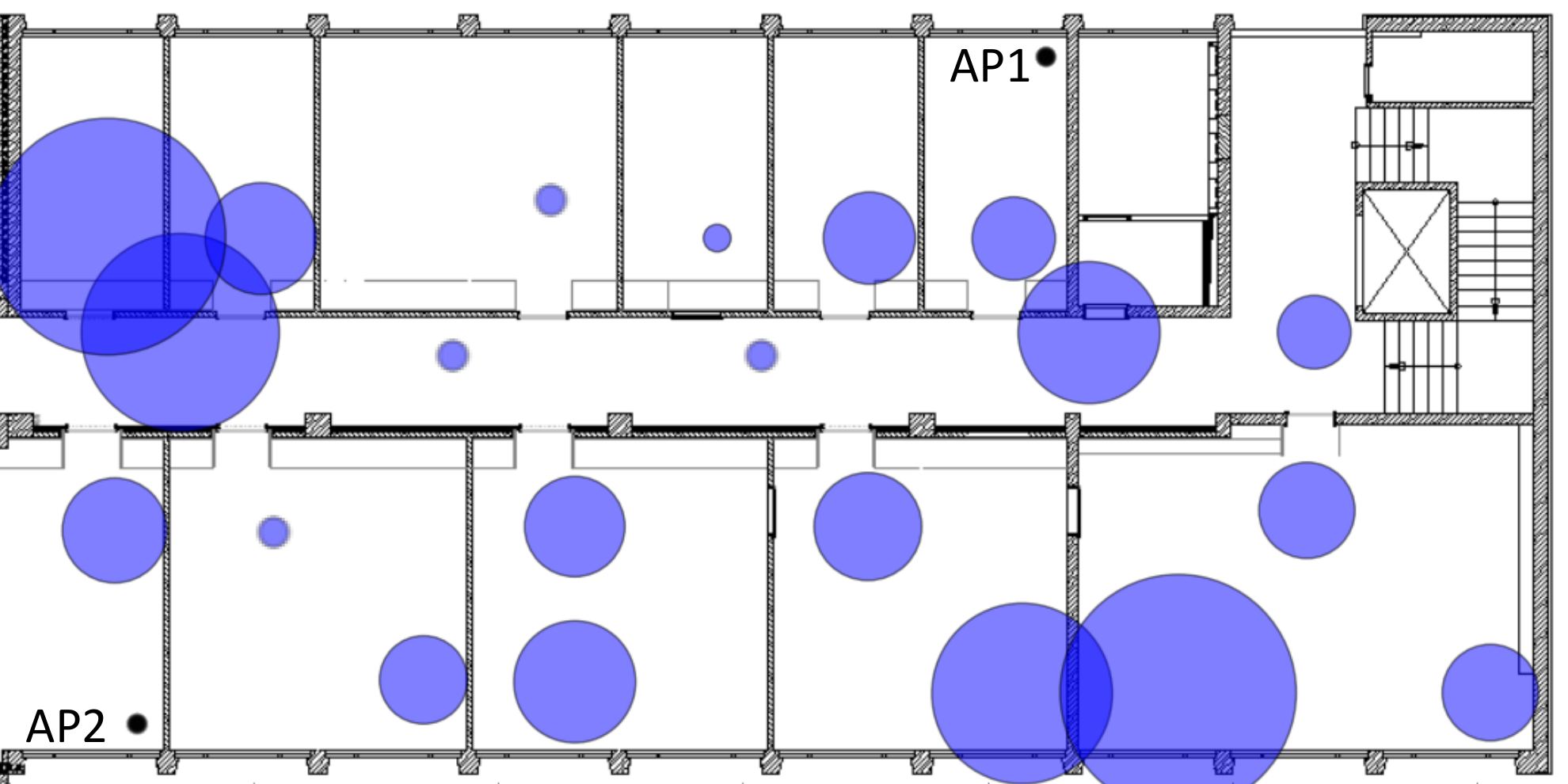}
\caption{Fingerprinting with AP1 and AP2}
\label{fig:errordist2}
\end{subfigure} \hspace{1mm}
\begin{subfigure}{0.32\linewidth}
\centering
\includegraphics[width=\columnwidth]{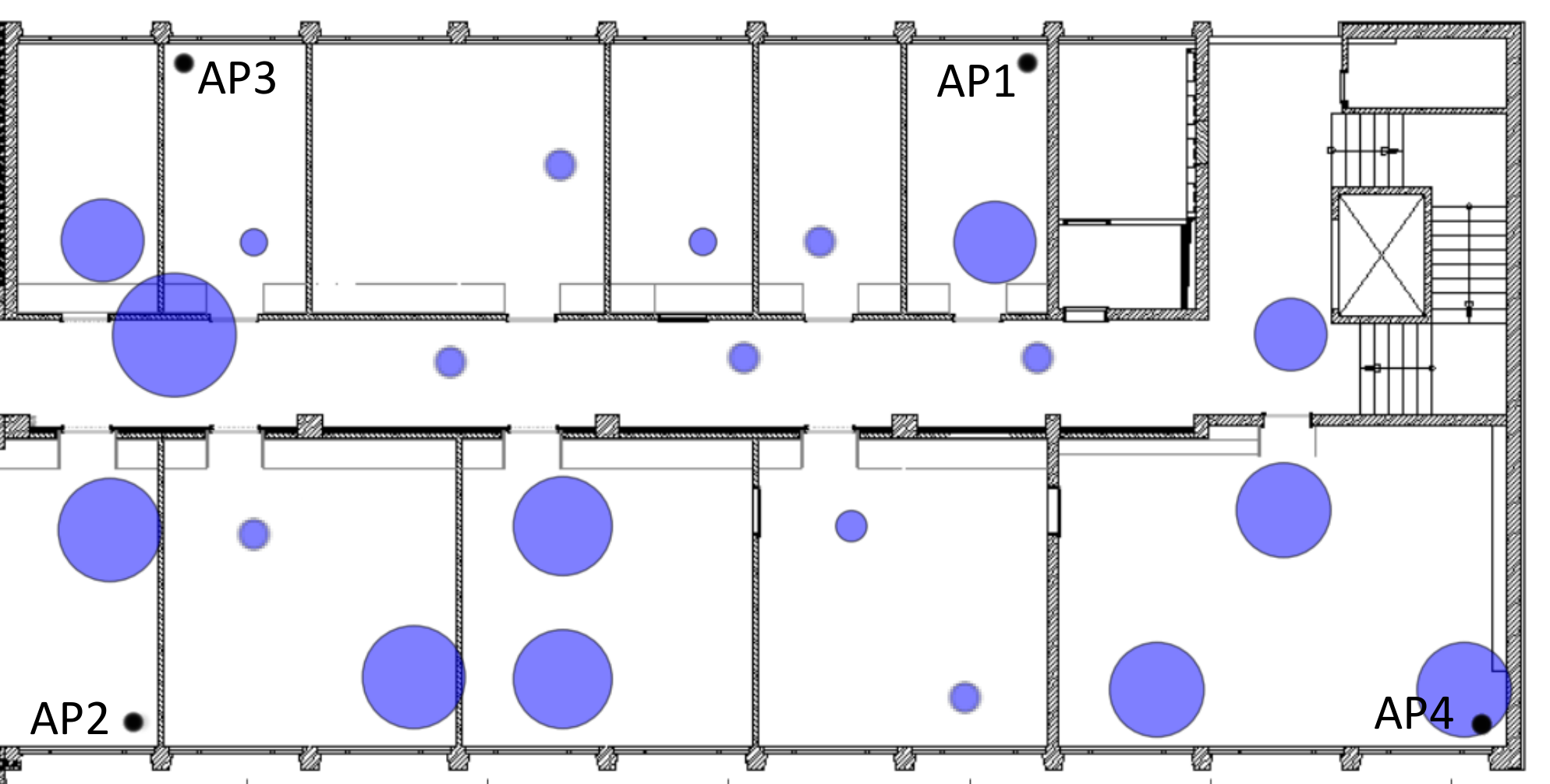}
\caption{Fingerprinting with all APs}
\label{fig:errordist3}
\end{subfigure}
\vspace{-1mm}
\caption{Spatial distribution of errors with different AP selection}
\label{fig:errordist}
\vspace{-3mm}
\end{figure*}

\subsection{Evaluation through Experimentation}
The TWIST testbed is specifically designed for indoor localization performance evaluation purposed experimentation.
It offers an automated experimentation without a need of a test-person, accurate ground-truth positioning, minimization and monitoring of external influences such as interference, and immediate calculation of performance results~\cite{lemic_infrastructure} aligned with the EVARILOS Benchmarking Handbook (EBH)~\cite{Van_Haute13the_evarilos_benchmarking}. 
The EBP provides guidelines for objective evaluation of RF-based indoor localization solutions.
In the TWIST testbed, various data-traces have been collected and offered to the public with the general goal of evaluation of RF-based indoor localization solutions~\cite{lemic15_platform_evaluation}.
These data-traces can be used in a streamlined fashion as an input to a localization solution to be evaluated. 
Based on this input location estimates are generated by the evaluated solution.
In the next step, these location estimates are compared with ground-truth locations and a set of metrics characterizing the performance of the evaluated solution in obtained.

A training database of the fingerprinting algorithm to be evaluated has been created by leveraging one of the available data-traces. 
The used data-trace has been generated by collecting 20~\ac{RSS} values from four \acp{AP} in 41 training locations, as indicated in Figure~\ref{fig:training}a).
It has been shown in the previous section that leveraging a hexagonal grid yields better performance of fingerprinting in comparison to a random or squared grid. 
Since the performance difference between hexagonal and square grids is not significant, in all currently available data-traces the measurements were collected in a squared-like fashion since this was practically more convenient.  
Finally, the evaluation locations used in the evaluation are also shown in Figure~\ref{fig:training}, with their selection based on the guidelines from the EBP.

First we evaluate the claim saying that the main benefit of using multiple \acp{AP} is the reduction of the effect of far APs.
In other words, the localization errors increase with the increase in the distance between an \ac{AP} and a target node.  
Figure~\ref{fig:errordist} presents spatial distribution of errors for three scenarios, each using different number of \acp{AP} for localization. 
As visible in Figure~\ref{fig:errordist1}, those locations farther away from AP1 tend to have larger errors. 
This is because the inverse relation of RSS values with distance such that those points closer to AP1 have a finer RSS granularity. 
It is interesting to see that fingerprinting localization algorithm performs acceptable in indoor environment even with one anchor due to shadowing and multi-path effects.
Furthermore, the \acp{AP} are added accordance to the guidelines discussed in the paper, i.e. first AP2 is deployed at the farthest location from AP1 following Voronoi-based deployment guideline. 
It can be seen in Figure~\ref{fig:errordist2} that such placement mainly decreases the errors at locations close to AP2. 
Second, in Figure~\ref{fig:errordist3}, four APs are deployed at four corners of the testbed, which again improves the localization errors. 
Both average and maximum geometric error is improved in this way as it can be seen in Figure~\ref{fig:apnr}.

\begin{figure}[!th]
\vspace{-2mm}
\centering
\includegraphics[width=0.6\columnwidth]{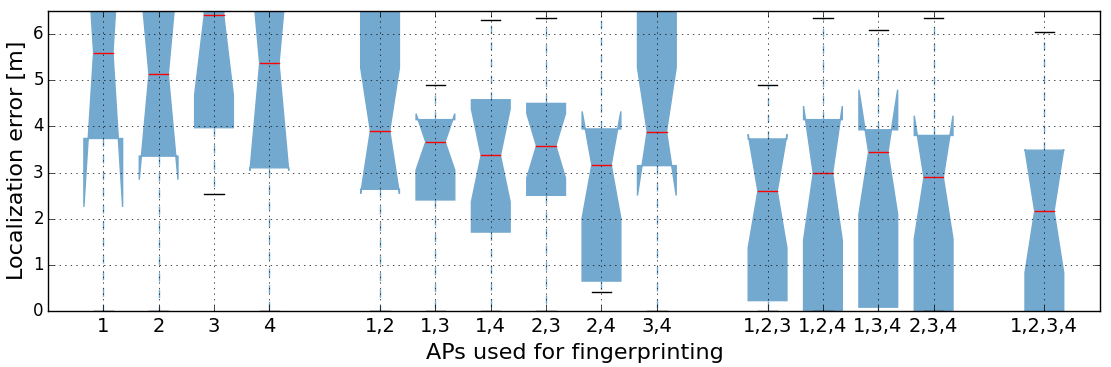}
\vspace{-6mm}
\caption{Localization error vs. number of APs} 
\label{fig:apnr}
\vspace{-5mm}
\end{figure}
\vspace{-1mm}
\section{Conclusion}
\label{conclusion}

Fingerprinting algorithms are generally multi-faceted and complex, hence it is hard to directly study the effect of their operational blocks and environmental parameters on their performance. 
In this work, a theoretical framework is introduced for analysis of fingerprinting algorithms with the goal of providing guidelines on their design and performance analysis. 
The performance of fingerprinting is shown to depend on the used fingerprinting feature and the dependence of this feature on the location in a fingerprinting space. 
A connection between fingerprinting and hypothesis testing problem was established. 
It was shown that the accuracy of fingerprinting algorithms is related to \ac{KL} divergence between probability distribution of the selected feature at two locations. It was suggested that \ac{KL} divergence can be used as a performance metric embedding both accuracy and latency of fingerprinting localizations.
The effect of the number of measurements at each location, as well as the effect of the training grid constellations on the peformance was discussed.
This framework was instantiated for RSS-based fingerprinting algorithms where the consistency of the introduced theoretical framework was examined. It has been shown that shadowing and fading act in favor of fingerprinting algorithms by creating more variability in fingerprints of different locations. Numerical simulations and experimental study confirm the claims of the theory. The framework is promising for considering various research problems in fingerprinting scenarios.

\vspace{-1mm}
\section*{Appendix}

\subsection{Proof of Theorem \ref{thm:1}}

Consider two points $\mathbf{u}_1$ and $\mathbf{u}_2$ in localization space. The theorem bounds missed detection and false alarm probabilities in a hypothesis testing scenario. The equivalent hypothesis testing problem is Stein's lemma and the proof is well known (for instance ~\cite{ dembo_large_2010, Csiszar1982}). Nonetheless the sketch of proof is provided here. Let the probabilities $\mathbb{P}_{X|\mathbf{u}_1}$ and $\mathbb{P}_{X|\mathbf{u}_2}$ be known. Consider the following set, called  \textit{typical} set:

\small
\begin{equation}
\begin{split}
 &A_{\epsilon}^n(X)=  \left\{\mb x\in\mathcal{X}^n: \left|\displaystyle\frac{1}{n}\sum_{i=1}^n\log\frac{\mathbb{P}_{X|\mathbf{u}_2}(x_i)}{\mathbb{P}_{X|\mathbf{u}_1}(x_i)}+ \right. \right. 
\left. \left. \vphantom{\frac{1}{n}} D(\mathbb{P}_{X|\mathbf{u}_1}\|\mathbb{P}_{X|\mathbf{u}_2})\right|\leq \epsilon\right\}.
\label{def:typset}
\end{split}
\end{equation}
\normalsize

Suppose that the fingerprint is constructed at $\mb u_1$. The law of large numbers implies that the term $\displaystyle\frac{1}{n}\sum_{i=1}^n\log\frac{\mathbb{P}_{X|\mathbf{u}_2}(x_i)}{\mathbb{P}_{X|\mathbf{u}_1}(x_i)}$ converges in probability (and also almost surely) to the expected value of $\frac{\mathbb{P}_{X|\mathbf{u}_2}(x)}{\mathbb{P}_{X|\mathbf{u}_1}(x)}$ as $n$ goes to infinity, which is $-D(\mathbb{P}_{X|\mathbf{u}_1}\|\mathbb{P}_{X|\mathbf{u}_2})$. 
Therefore for large enough $n$, it can be seen that:
\begin{equation}
\alpha(\mathbf{u}_1,\mathbf{u}_2)=\mathbb{P}_{\mb X|\mathbf{u}_1}(\mb X \notin A_{\epsilon}^n(X)) \leq \epsilon.
\end{equation}
This means that if the fingerprint is recorded at location $\mb u_1$, it will belong to $A_{\epsilon}^n(X)$ with high probability and hence it can act as a decision region for  $\mathbb{P}_{X|\mathbf{u}_1}$. The preceding inequality guarantees that the correct identification with probability bigger than $1-\epsilon$. It remains to show that if the fingerprint is recorded at $\mb u_2$, it will not belong to $A_{\epsilon}^n(X)$. To this purpose,  $\beta(\mathbf{u}_1,\mathbf{u}_2)$ should be bounded. Suppose that the samples are obtained at $\mathbf{u}_2$: 

\begin{equation}
\begin{split} 
\beta(\mathbf{u}_1,\mathbf{u}_2)&=\mathbb{P}_{\mb X|\mathbf{u}_2}(\mb X\in A_{\epsilon}^n(X)) = \displaystyle\sum_{\mb x\in A_{\epsilon}^n(X)}\mathbb{P}_{\mb X|\mathbf{u}_2}(\mb x) \\
&\stackrel{(a)}{\leq}\displaystyle\sum_{\mb x\in A_{\epsilon}^n(X)}\mathbb{P}_{\mb X|\mathbf{u}_1}(\mb x)\exp(-n(D(\mathbb{P}_{X|\mathbf{u}_1}\|\mathbb{P}_{X|\mathbf{u}_2})-\epsilon))\\
&\displaystyle\leq \exp(-n(D(\mathbb{P}_{X|\mathbf{u}_1}\|\mathbb{P}_{X|\mathbf{u}_2})-\epsilon)).
\end{split}
\end{equation}
where $(a)$ follows from the definition \eqref{def:typset}. This shows that $\beta(\mathbf{u}_1,\mathbf{u}_2)$ tends to zero as $n\to\infty$ and therefore $A_{\epsilon}^n(X)$ is good decision region for distinguishing $\mb u_1$ and $\mb u_2$. Note that $\beta(\mathbf{u}_1,\mathbf{u}_2)$ converges to zero exponentially with $n$. It is possible to find this exponent. Having an upper bound on $\beta(\mathbf{u}_1,\mathbf{u}_2)$, the inner bound is obtained as follows:
\begin{equation}
\begin{split}
\beta(\mathbf{u}_1,\mathbf{u}_2)&=\mathbb{P}_{\mb X|\mathbf{u}_2}(\mb X\in A_{\epsilon}^n(X)) = \displaystyle\sum_{\mb x\in A_{\epsilon}^n(X)}\mathbb{P}_{\mb X|\mathbf{u}_2}(\mb x)  \\
&\displaystyle\geq\sum_{\mb x\in A_{\epsilon}^n(X)}\mathbb{P}_{\mb X|\mathbf{u}_1}(\mb x)\exp(-n(D(\mathbb{P}_{X|\mathbf{u}_1}\|\mathbb{P}_{X|\mathbf{u}_2})+\epsilon))\\
&\displaystyle= \alpha(\mathbf{u}_1,\mathbf{u}_2)\exp(-n(D(\mathbb{P}_{X|\mathbf{u}_1}\|\mathbb{P}_{X|\mathbf{u}_2})+\epsilon))\\
&\displaystyle\geq (1-\epsilon)\exp(-n(D(\mathbb{P}_{X|\mathbf{u}_1}\|\mathbb{P}_{X|\mathbf{u}_2})+\epsilon)).
\end{split}
\end{equation}
By taking the logarithm from inner and upper bounds on $\beta(\mathbf{u}_1,\mathbf{u}_2)$ and tending $n$ to infinity, the exponent is proved to be $D(\mathbb{P}_{X|\mathbf{u}_1}\|\mathbb{P}_{X|\mathbf{u}_2})$.

\subsection{Proof of Theorem \ref{thm:1.1}}
Consider the set $C_n$, called the critical region, defined as follows:
$$
C_n=\left\{\mb x: \inf  \inf_{\mb u_1 \in \hat{\mathcal{V}}_{\mathbf{v}_1}} D(\mathbb{Q}_{X|\mb u}\| \mathbb{P}_{X|{\mathbf{u}_1}}) \geq \delta_n \right\}
$$
where $\mathbb{Q}_{X|\mb u}$ is the empirical distribution of $\mb x$ measured at $\mb u$ and $\delta_n=\Omega(\frac{\log n}{n})$. Then using Theorem 2.3 in \cite{csiszar_information_2004}, it can be seen that if $\mb X$ follows the distribution  $\mathbb{Q}_{X|\mb u}$ for $\mb u\in \hat{\mathcal{V}}_{\mathbf{v}_1}$, then $\Pr(\mb X\in C_n)\leq \epsilon$. Using the same theorem, the second error is shown to decay exponentially with the exponent indicated in the theorem. So  $(C_n)^s$ can be a decision region for finding the closest fingerprint.

\subsection{Proof of Proposition \ref{prop:stab}}

Suppose that $D(\mathbb{P}_{\mb X|\mathbf{u}_1}\| \mathbb{P}_{\mb X|\mathbf{u}_2})\leq L$. Using \eqref{eq:QNoiseModel} for an anchor in origin, and the points at $u_1>0$ and $u_2>0$, we have:
$$
D(\mathbb{P}_{\mb X|\mathbf{u}_1}\| \mathbb{P}_{\mb X|\mathbf{u}_2})=\frac{P^2_T}{2N_1}\left(\frac{1}{{u}_1^{\alpha}}-\frac{1}{{u}_2^{\alpha}}\right)^2.
$$
A simple usage of mean value theorem implies that for some $\beta\in[0,1]$:
$$
\left|\frac{1}{{u_1}^\alpha}-\frac{1}{{u_2}^\alpha}\right |= |{ u_1- u_2}| \left| \frac{\alpha}{({\beta  u_1+(1-\beta) u_2})^{\alpha+1}}\right|.
$$
For this case if $D( \prob_{\mb X|\mb u_1}\| \prob_{\mb X|\mb u_2}) = L$, then:
$$
|{ u_1- u_2}|^2\leq  \frac{L}{\kappa} {({\beta  u_1+(1-\beta) u_2})^{2\alpha+2}}
$$
where $\kappa=\frac{\alpha^2P^2_T}{2N_1}$. 
Since the localization space is bounded inside $[0,D]$ then $\beta  u_1+(1-\beta) u_2\leq D$ and therefore:
$$
 |{ u_1- u_2}|\leq \frac{D^{\alpha+1}}{\alpha P_T} \sqrt{2N_1 L}.
$$

\vspace{-4mm}
\bibliographystyle{IEEEtran}
{\footnotesize \bibliography{biblio}}

\end{document}